%% file: gcm23.tex
\documentclass[submission,copyright,creativecommons, colorlinks]{eptcs}
\providecommand{\event}{GCM 2023} 
\usepackage[T1]{fontenc} 
\usepackage{lmodern}
\usepackage[utf8]{inputenc}
\usepackage[inline]{enumitem}
	\setlist{itemsep=0pt,topsep=.5ex}
\usepackage{xspace}
\usepackage{booktabs}
\usepackage{tikz}                     
  \usetikzlibrary{automata,matrix,calc,chains,backgrounds,positioning
                 ,fit,decorations.pathmorphing,shapes
                 ,arrows,arrows.meta}

\usepackage{amstext}
\usepackage{amsthm} 
\usepackage{amsmath}
\usepackage{amssymb}
\usepackage{stmaryrd}
\usepackage{wasysym}
\usepackage{mathtools}
\usepackage[ruled,vlined,linesnumbered]{algorithm2e}

\newcommand{\tableref}[1]{\hyperref[#1]{Table~\ref*{#1}}}
\newcommand{\figref}[1]{\hyperref[#1]{Figure~\ref*{#1}}}
\newcommand{\figandfigref}[2]%
	{\hyperref[#1]{Figures~\ref*{#1}} \hyperref[#1]{and~\ref*{#2}}}
\newcommand{\assref}[2]%
	{\hyperref[#1]{Assumption~\ref*{#1}}.\hyperref[#1]{\ref*{#2}}}
\newcommand{\assrefs}[3]%
	{\hyperref[#1]{Assumptions~\ref*{#1}}.\hyperref[#1]{\ref*{#2}}--\hyperref[#1]{\ref*{#1}}.\hyperref[#1]{\ref*{#3}}}
\newcommand{\tabref}[1]{\hyperref[#1]{Table~\ref*{#1}}}
\newcommand{\sectref}[1]{\hyperref[#1]{Section~\ref*{#1}}}
\newcommand{\sectsref}[1]{\hyperref[#1]{Sections~\ref*{#1}}}
\newcommand{\stepref}[1]{\hyperref[#1]{Step~\ref*{#1}}}
\newcommand{\defref}[1]{\hyperref[#1]{Definition~\ref*{#1}}}
\newcommand{\thmref}[1]{\hyperref[#1]{Theorem~\ref*{#1}}}
\newcommand{\lemmaref}[1]{\hyperref[#1]{Lemma~\ref*{#1}}}
\newcommand{\factref}[1]{\hyperref[#1]{Fact~\ref*{#1}}}
\newcommand{\factsubref}[2]%
	{\hyperref[#1]{Fact~\ref*{#1}}.\hyperref[#1]{\ref*{#2}}}
\newcommand{\exref}[1]{\hyperref[#1]{Example~\ref*{#1}}}
\newcommand{\exsref}[1]{\hyperref[#1]{Examples~\ref*{#1}}}
\newcommand{\algref}[1]{\hyperref[#1]{Algorithm~\ref*{#1}}}
\newcommand{\alineref}[1]{\hyperref[#1]{Line~\ref*{#1}}}
\newcommand{\alinesref}[1]{\hyperref[#1]{Lines~\ref*{#1}}}
\newcommand{\remref}[1]{\hyperref[#1]{Remark~\ref*{#1}}}
\newcommand{\condref}[1]{\hyperref[#1]{Condition~\ref*{#1}}}
\newcommand{\Eqref}[1]{\hyperref[#1]{Equation~\ref*{#1}}}
\newcommand{\constrref}[1]{\hyperref[#1]{Construction~\ref*{#1}}}
\newcommand{\obsref}[1]{\hyperref[#1]{Observation~\ref*{#1}}}
\newcommand{\corref}[1]{\hyperref[#1]{Corollary~\ref*{#1}}}

\input{myeptcstheoremdefs.tex}  
\input{mymath.tex}         
\input{mytikz.tex}         

\input{mygraphs.tex}

\definecolor{lime}{HTML}{A6CE39}
\DeclareRobustCommand{\orcidicon}{
	\begin{tikzpicture}
	\draw[lime, fill=lime] (0,0) 
	circle [radius=0.16] 
	node[white] {{\fontfamily{qag}\selectfont \tiny ID}};
	\draw[white, fill=white] (-0.0625,0.095) 
	circle [radius=0.007];
	\end{tikzpicture}
	\hspace{-2mm}
}
\foreach \x in {D,H,M}{\expandafter\xdef\csname orcid\x\endcsname{\noexpand\href{https://orcid.org/\csname orcidauthor\x\endcsname}
			{\noexpand\orcidicon}}
}

\begin{document}
\title{Finite Automata for Efficient Graph Recognition}
\author{Frank Drewes\!\!\!\orcidD{}%
  \institute{
    Ume\aa\ universitet \\
    SE-90187 Ume\aa \\
    Sweden}
    \email{drewes@cs.umu.se}
  \and Berthold Hoffmann
  \!\!\!\orcidH{}
  \institute{
    Universität Bremen\\
    D-28334 Bremen \\
    Germany}
    \email{hof@uni-bremen.de}
  \and Mark Minas\!\!\!\orcidM{}
  \institute{
    Universität der Bundeswehr München \\
    D-85577 Neubiberg \\
    Germany}
    \email{mark.minas@unibw.de}
}
\def\authorrunning{F. Drewes, B. Hoffmann, M. Minas}
\def \titlerunning {{\sl Finite automata for efficient graph recognition}}
\maketitle
\begin{abstract}
  Engelfriet and Vereijken have shown that linear graph grammars based
  on hyperedge replacement generate graph languages that can be
  considered as interpretations of regular string languages over
  typed symbols.
  In this paper we show that finite automata can be lifted from
  strings to graphs within the same framework.
  For the efficient recognition of graphs with these automata, we make
  them deterministic by a modified powerset construction, and
  state sufficient conditions under which deterministic finite graph
  automata recognize graphs without the need to use backtracking.
\end{abstract}

\section{Introduction}
\label{s:intro}

Engelfriet and Vereijken \cite{Engelfriet-Vereijken:97} have shown
that linear graph grammars based on hyperedge replacement can be
considered as interpretations of linear string grammars: typed (``doubly ranked'')
symbols of an alphabet are interpreted as basic graphs that have
front and rear interfaces of nodes, and string concatenation is
interpreted as the composition of two graphs by gluing the rear of the
first to the front of the second graph. Graph languages constructed in
this way are of bounded pathwidth, and are thus potentially more efficiently
recognizable than general hyperedge replacement languages, which are
known to be NP-complete. However, without additional restrictions even
these graph languages are NP-complete~\cite{Aalbersberg-Ehrenfeucht-Rozenberg:86}.

In this paper we study how finite automata over graph symbols can be
interpreted to recognize graph languages efficiently. Given a graph as
input, the transitions of such an automaton consume the graph step by
step while changing the state of the automaton, in the end reaching either
an accepting or a rejecting state. If the automaton used is nondeterministic,
a naive decision procedure for determining whether the input graph is
accepted would have to use backtracking. While we show in this paper that
the transition relation can be
made deterministic by a modified powerset construction, backtracking may
still be needed. This is due to the fact that, whereas a string starts
with a unique first symbol to be read by the first
transition of the computation of a deterministic finite automaton, in general several
alternative basic graphs represented by that symbol may be spelled off at the front of a graph.
We provide two sufficient criteria under which
automata can choose between them without the need to backtrack.
These criteria resemble similar criteria known from efficient parsing algorithms for
context-free hyperedge replacement languages
\cite{Drewes-Hoffmann-Minas:19a,Drewes-Hoffmann-Minas:19}.

Work on efficient parsing algorithms for  grammars
started in the late 1980s, initiated by the realization that, in general, the graph
languages generated by these grammars can be
NP-complete~\cite{Aalbersberg-Ehrenfeucht-Rozenberg:86,Lange-Welzl:87}. Early
polynomial algorithms were based on restrictions which either ensure efficiency
of the well-known Cocke-Younger-Kasami algorithm (adapted to hyperedge replacement
grammars) or make sure that the derivation trees of generated graphs mirror a
unique recursive decomposition of graphs into smaller and smaller subgraphs,
see~\cite{Lautemann:90,Vogler:91,Drewes:93c}. Later work on parsing hyperedge
replacement languages include~\cite{chiang-et-al:2013,Gilroy-Maneth-et-al:17}
as well as the authors' own work on top-down and bottom-up parsers for these languages; see, e.g.,~\cite{Drewes-Hoffmann-Minas:15,Drewes-Hoffmann-Minas:17,Hoffmann-Minas:18,Drewes-Hoffmann-Minas:19,Drewes-Hoffmann-Minas:21}.
For a more extensive overview of work on efficient parsing for graph grammars,
including other types of grammars than those based on hyperedge replacement,
see~\cite[Section~10]{Drewes-Hoffmann-Minas:19}.

For the efficient recognition of graph languages, finite
automata have not attained the importance that they
have for 
string languages.
Some early work exists on finite automata for algebraic structures
\cite{Arbib-Giveon:68,Giveon-Arbib:68}, rooted directed acyclic graphs
\cite{Witt:81}, and infinite directed acyclic graphs
\cite{Kaminski-Pinter:92}.  Brandenburg and Skodinis have devised 
finite automata for linear node replacement
grammars~\cite{Brandenburg-Skodinis:05}.  Bozapalidis and Kalampakas
have studied automata on the hypergraphs of Engelfriet and Vereijken
\cite{Engelfriet-Vereijken:97} in an algebraic setting
\cite{Bozapalidis-Kalampakas:08,Kalampakas:11}.  And, Brugging, König
\emph{et al.} have recast these hypergraphs as cospans 
\cite{Blume-Bruggink-Friedrich-Koenig:13}, and have shown that finite
automata defining hypergraph languages are equivalent to Courcelle's
recognizable hypergraph languages \cite{Courcelle:90}; later these
results have been generalized to hereditary pushout categories
\cite{Bruggink-Koenig:18}.
All this work has focussed on 
\emph{defining} the graph languages accepted by  finite automata (e.g., by composing
the graphs labeling the edges passed on a walk through the
transition diagrams of the automata), rather then on \emph{recognizing}
graph languages in the sense of efficiently deciding their membership problem.
In particular, the efficiency of the recognition process has not been a
matter of concern.
In this paper, we focus on this aspect, which is of interest because
the analysis of graphs by finite automata
corresponds intuitively to a lexical analysis in the traditional string-based
setting.

The technical contributions of this paper revolve mainly around different
aspects of coping with nondeterminism, which differs in subtle ways from
the string case. Given a nondeterministic automaton over graph symbols,
we first apply a powerset construction to it (\sectref{sec:powerset automaton}).
Interestingly, this construction does not always result in a deterministic
automaton. We show, however, that the resulting automaton is indeed deterministic if the
construction starts from an unambiguous automaton, a condition that can
always be fulfilled (\sectref{sec:deterministic}).
Next, as mentioned above, even a deterministic automaton is only deterministic
at the string level. As the edges of a graph do not come in a
predefined processing order, recognition still requires backtracking in
general: anytime in the process, several different transitions
may be applicable, and a given transition may have several alternative
edges it may consume. We deal with these problems in
\sectsref{sec:select} and~\ref{sec:FEC}, resulting in linear-time
recognition without backtracking.


The remainder of this paper is structured as follows.
We define graphs with front and rear interfaces and their
composition in \sectref{s:lgl}.
In \sectref{s:fga}, we introduce finite automata over graph symbols
and define how they can be used to recognize graph languages.
\sectref{s:pfga} contains the main technical contributions, as
described above.
In \sectref{s:concl}, we 
conclude the paper and
indicate directions of future research.

\section{Graphs and Graph Composition}
\label{s:lgl}

We let $\Nat$ denote the set of non-negative integers, and $[n]$ the
set $\{1,\dots,n\}$ for all $n\in\Nat$. $A^*$ denotes the set of all
finite sequences over a set~$A$; the empty sequence is denoted by
$\emptyseq$, $A^\ast\setminus\{\emptyseq\}$ by $A^+$, and the length of a 
sequence $\alpha$ by $|\alpha|$. For a sequence $s=a_1\cdots a_n$ with
$n\in\Nat$ and $a_1,\dots,a_n \in A$, we let $[s]=\{a_1,\ldots,a_n\}$ 
and $s(i)=a_i$ for all $n\in\Nat$. For $f\colon A\to B^*$, we let
$f(a,i)=f(a)(i)$ if $a\in A$ and $i\in[|f(a)|]$.
We silently extend functions to sequences and let 
$f(a_1 \cdots a_n) = f(a_1) \cdots f(a_n)$, for all $n\in\Nat$,
$a_1,\dots,a_n \in A$, and all functions $f\colon A \to B$.
As usual, given a binary relation $R\subseteq A\times B$, we denote
the transitive closure by $R^+$ and the transitive reflexive closure by 
${R^*}$. Further, $\domain R$ denotes the domain of $R$, i.e.,
$\domain R=\{a\in A\mid\exists b\in B\colon (a,b)\in R\}$. Finally, we let
$R(a)=\{b\in B\mid (a,b)\in R\}$ for every $a\in A$.

We consider edge-labeled hypergraphs (which we simply call graphs).
Like Habel in \cite{Habel:92}, we
supply them with a front and a rear interface, each being a sequence of
nodes.\footnote{Other than in \cite{Habel:92}, however, we do not
  divide the attached nodes of edges into sources and
  targets.}

For the labeling of edges,
we consider a ranked alphabet $(\Voc,\rank)$ consisting of
a set of symbols and a function
$\rank \colon \Voc \to \Nat$ which assigns
a rank to each symbol $a\in\Voc$. The pair $(\Voc,\rank)$
is usually identified with $\Voc$, keeping $\rank$ implicit. 

\begin{defn}[Graph]\label{s:graph}%
  A \emph{graph} (over $\Voc$) is a tuple
  $G = (\nd G, \ed G, \att_G, \allowbreak \lab_G, \allowbreak
    \front_G, \allowbreak \rear_G)$, where $\nd G$ and $\ed G$ are
  disjoint finite sets of \emph{nodes} and
  \emph{edges}, respectively, the function
  $\att_G \colon \ed G \to \nd G ^*$ attaches repetition-free sequences of nodes to
  edges, the function $\lab_G \colon G \to \Voc$ labels edges with
  symbols in such a way that $\card{\att_G(e)} = \rank(\lab_G(e))$ for
  every edge $e \in \ed G$, and the repetition-free node sequences
  $\front_G, \rear_G \in \nd G^*$ specify the \emph{front} and
  \emph{rear} interface nodes.

  The lengths of front and rear interfaces classify graphs as
  follows: A graph $G$ \emph{has type $(i,j)$} if
  $\card{\front_G} = i$ and $\card{\rear_G} = j$; we then write $\type(G)=(i,j)$.
  $\GG_\Voc^{(i,j)}$ denotes the set of all graphs of type $(i,j)$, and
  $\GG_\Voc = \bigcup_{i,j \in\Nat} \GG_\Voc^{(i,j)}$ denotes the set
  of all graphs over $\Voc$, regardless of type.
\end{defn}

For graphs $G$ and $H$, a \emph{morphism} $m \colon G \to H$ is a pair
$m=(\nd m, \ed m)$ of functions $\nd m \colon \nd G \to \nd H$ and
$\ed m \colon \ed G \to \ed H$ that preserve attachments and labels,
i.e., $\att_H(\ed m(e)) = \nd m(\att_G(e))$ and %
$\lab_H(\ed m(e)) = \lab_G(e)$ for all  $e\in\ed G$. (Note that fronts
and rears need not be preserved.)
The morphism $m$ is \emph{injective} or \emph{surjective} if both
$\nd m$ and $\ed m$ have this property, and a \emph{subgraph
  inclusion} of $G$ in $H$ if $m (x) = x$ for every node and edge $x$
in $G$; then we write $G \subseteq H$.
$G$ and $H$ are \emph{isomorphic}, $G \iso H$, if there exists a surjective
and injective morphism $m\colon G \to H$ such that $\front_H = m(\front_G)$
and $\rear_H = m(\rear_G)$.

Consider $a\in \Voc$, $n\in\Nat$, and repetition-free 
sequences $\phi,\rho\in[n]^\ast$ such that $[\phi]\cup[\rank(a)]=[n]$
(where $[\phi]$ and $[\rank(a)]$ are not necessarily disjoint). 
Then $\tup{a}^\phi_\rho$ denotes
the graph $A$ with $\nd A=[n]$, $\ed A=\{e\}$, $\lab_A(e)=a$, $\att_A(e)=1\cdots \rank(a)$, 
$\front_A=\phi$, and $\rear_A=\rho$, whereas $\tup{\emptyseq}^{(n)}_\rho$ 
denotes the discrete graph $B$ with $\ed B=\emptyset$, $\nd B=[n]$, 
$\front_B=1\cdots n$, and $\rear_A=\rho$.
We call $\tup{a}^\phi_\rho$ an \emph{atom} and $\tup{\emptyseq}^{(n)}_\rho$ a \emph{blank}.

Note that $\tup{\emptyseq}^{(0)}_\emptyseq$ is the empty graph.
Further note that no atom has any node that neither is a front node, 
nor is it attached to its only edge. In particular, every rear node is also a front 
node or attached to the unique edge of the atom, or both. Moreover, all nodes of a blank 
occur in its front interface. Finally recall that our objective is efficient
recognition of graphs, which tries to compose an input graph from a sequence of atoms 
and a blank. The requirements on front and rear interfaces differ because recognition 
will process graphs from front to rear.

\begin{figure}[tb]
\centering
  \begin{tabular}{@{}c@{}}
    \begin{graph}[x=10mm,y=-10mm]
      \fnode(w)(1,1) 
      \inode(x)(1,0) 
      \rnode(y)(2,1) 
      \inode(z)(1,2) 
      \path
      (w) edge[->] node[left] {\small$a$} (x)
      (w) edge[->] node[above] {\small$b$} (y)
      (w) edge[->] node[left] {\small$a$} (z)
      ;
    \end{graph} \\\\[-2mm]
    ``Star'' graph
  \end{tabular}
  \qquad
  \begin{tabular}{@{}c@{}}
    \begin{graph}[x=10mm,y=-10mm]
      \frnode(1)(1,1) 
      \rnode(2)(1,3) 
      \path
      (1) edge[->] node[right] {\small$b$} (2)
      ;
    \end{graph} \\\\[-2mm]
    $\gatom{b}{1}{12}$
  \end{tabular}
 \qquad
  \begin{tabular}{@{}c@{}}
   \begin{graph}[x=10mm,y=-10mm]
      \frnode(1)(1,1) 
      \rnode(2)(1,2.5) 
      \frnode(3)(1,3) 
      \path
      (1) edge[->] node[right] {\small$a$} (2)
      ;
    \end{graph} \\\\[-2mm]
    $\gatom{a}{13}{123}$
  \end{tabular}
  \qquad
  \begin{tabular}{@{}c@{}}
   \begin{graph}[x=10mm,y=-10mm]
      \frnode(1)(1,1) 
      \rnode(2)(1,2) 
      \frnode(3)(1,2.5) 
      \frnode(4)(1,3) 
      \path
      (1) edge[->] node[right] {\small$a$} (2)
      ;
    \end{graph} \\\\[-2mm]
    $\gatom{a}{134}{1234}$
  \end{tabular}
  \qquad
  \begin{tabular}{@{}c@{}}
   \begin{graph}[x=10mm,y=-10mm]
      \fnode(1)(1,1) 
      \rnode(2)(1,2) 
      \frnode(3)(1,2.5) 
      \frnode(4)(1,3) 
      \path
      (1) edge[->] node[right] {\small$a$} (2)
      ;
    \end{graph} \\\\[-2mm]
    $\gatom{a}{134}{234}$
  \end{tabular}
  \qquad
  \begin{tabular}{@{}c@{}}
   \begin{graph}[x=10mm,y=-10mm]
      \fnode(1)(1,1) 
      \fnode(2)(1,2) 
      \fnode(3)(1,2.5) 
      \frnode(4)(1,3) 
    \end{graph} \\\\[-2mm]
    $\gblank{4}{4}$
  \end{tabular}
  \caption{A ``star'' graph, four atoms, and a blank.}
  \label{f:atoms-blank}
\end{figure}
\begin{example}\label{ex:atoms-blank}%
  \figref{f:atoms-blank} shows a ``star'' graph, atoms
  $\gatom{b}{1}{12}$, $\gatom{a}{13}{123}$, $\gatom{a}{134}{1234}$, 
  and $\gatom{a}{134}{234}$
  using the symbols $a$ and $b$, both of rank~$2$, and the blank
  $\gblank{4}{4}$.  As usual, we represent nodes by circles and binary
  edges by arrows.  Front interface nodes are connected with double
  lines to the left border of the graph, rear interface nodes with
  double lines to its right border.
  Front and rear nodes are ordered from top to bottom;
  the ``star'' graph, e.g., has
  the front interface~$w$ and the rear interface~$y$.
\end{example}

We follow Engelfriet and Vereijken \cite{Engelfriet-Vereijken:97} in
defining the composition of graphs in a way resembling string
concatenation.

\begin{defn}[Graph Composition]\label{d:graphops}
  Let $G \in \GG_\Voc^{(i,k)}$ and $H \in \GG_\Voc^{(k,j)}$.
  We assume for simplicity that $\rear_G=\front_H$, $\nd G\cap\nd H=[\rear_G]=[\front_H]$, 
  and $\ed G\cap\ed H=\emptyset$ (otherwise an appropriate isomorphic copy of $G$ or $H$ is used).
  The (\emph{typed}) \emph{composition} $G \gcomp H$ of $G$ and $H$ is the graph $C$ such that
  $\nd C=\nd G\cup\nd H$, $\ed C=\ed G\cup\ed H$, $\att_C=\att_G\cup\att_H$, 
  $\lab_C=\lab_G\cup\lab_H$, $\front_C=\front_G$, and $\rear_C=\rear_H$. Thus $C \in \GG_\Voc^{(i,j)}$.
\end{defn}

Note that the composition $G\gcomp H$ is defined on concrete graphs if
the assumptions in \defref{d:graphops} are satisfied, but is only defined up
to isomorphism if an isomorphic copy of $G$ or $H$ needs to be taken. To avoid
unnecessary technicalities, we shall assume that these assumptions
are indeed satisfied whenever convenient.

Graphs can be constructed from a finite set of basic graphs using
composition and disjoint union
\cite{Gadducci-Heckel:97,Bozapalidis-Kalampakas:06} (where the latter
concatenates the fronts and rears of the two graphs involved).
If we just use composition, we need finitely many basic graphs per
pair of front and rear interfaces. Instead of the simpler  ``atomic graphs''
proposed in \cite{Engelfriet-Vereijken:97} and by Blume et al.{} in
\cite{Blume-Bruggink-Friedrich-Koenig:13}, we use atoms and
blanks here.

\begin{lemma}
\label{l:graspel}
  Every graph $G$ whose isolated nodes all occur in its front interface is of the form
  $A_1\gcomp\cdots\gcomp A_n\gcomp B$ where $n=|\ed G|$, $A_1,\dots,A_n$ are atoms,
  and $B$ is a blank.
\end{lemma}
\begin{proof}
  Let $G \in \GG_\Voc$ with $|\ed G|=n$ such that every isolated 
  node of $G$ occurs in $[\front_G]$. We prove by induction on $n$ that there are atoms 
  $A_1,\ldots,A_n$ and a blank $B$ such that $G\iso A_1 \gcomp\cdots\gcomp A_n \gcomp B$.
  If $n=0$, all nodes of $G$ are isolated, and hence $\nd G=[\front_G]$. 
  Then $G\iso\gblank{m}{\rho}$ where $m=|\nd G|$ and $\rho$ is a sequence of length $k=|\rear_G|$ 
  such that $\rear_G(i)=\front_G(\rho(i))$ for all $i\in[k]$.

  Now consider $n>0$ and assume, as an induction hypothesis, that the proposition 
  holds for all graphs with $n-1$ edges. Pick any edge of $G$, say $e\in\ed G$,
  and let $G'$ be the graph obtained from $G$ by removing $e$ and setting $\front_{G'}$ 
  to any permutation of $[\front_G]\cup[\att_G(e)]$. Since $G'$ has $n-1$ edges, 
  there are $n-1$ atoms $A_2,\ldots,A_n$ and a blank $B$ such that 
  $G' \iso A_2 \gcomp\cdots\gcomp A_n \gcomp B$ by the induction hypothesis.
  Now let $m=|\front_G|$, $k=|\front_{G'}|$, and $A_1=\gatom{\lab_G(e)}{\phi}{\rho}$ 
  where $\rho$ is a permutation of $[k]$, and $\phi$ is a sequence of length $m$ 
  such that the following holds: $\front_{G'}(i)=\att_G(e,\rho(i))$ for all $i\in[k]$ 
  with $\front_{G'}(i)\in[\att_G(e)]$. Moreover, $\phi(i)=\rho(j)$ if
  $\front_G(i)=\front_{G'}(j)$, for all $i\in[m]$ and $j\in[k]$. 
  It is easy to see that 
  $G\iso A_1\gcomp G'\iso A_1\gcomp\cdots\gcomp A_n\gcomp B$.
\end{proof}

Note that the size of interfaces required to build $G$ according to \lemmaref{l:graspel}
depends on $G$. Hence, the lemma does not contradict the fact, mentioned in
\sectref{s:intro}, that the set of all graphs
composed from a given finite set of atoms and blanks is of bounded pathwidth.

\begin{example}\label{ex:composition}%
	Following the construction in the proof, the ``star'' graph $G$ in 
	\figref{f:atoms-blank} can be composed by 
	$G\iso\gatom{b}{1}{12}\gcomp\gatom{a}{13}{123}\gcomp\gatom{a}{134}{1234}\gcomp\gblank{4}{4}$,
	all shown in \figref{f:atoms-blank}. However, this is not the only composition. 
	A much simpler one is $G\iso\gatom{a}{1}{1}\gcomp\gatom{a}{1}{1}\gcomp\gatom{b}{1}{2}$.
\end{example}

The example shows that, unlike a string, a graph can be composed from atoms and blanks
in different ways. Moreover, the size of atoms and blanks can also vary substantially.

\section{Finite Automata over Graph Symbols}
\label{s:fga}\label{s:grasymlang}

We now follow the idea of Engelfriet and Vereijken
\cite{Engelfriet-Vereijken:97} to make use of the close resemblance of graphs
under (typed) composition and strings under concatenation, and denote graphs by strings.
Each symbol is interpreted as a graph, and thus a string translates into a graph,
provided that the types of composed graphs fit. To this end, we type each symbol by a
pair $(i,j)$ -- the type of the graph it will represent.

A \emph{typed alphabet} is a pair $(\AllTerms,\type)$
consisting of a (possibly infinite) set $\AllTerms$ of symbols
and a function $\type\colon\AllTerms\to\Nat\times\Nat$ which assigns
a pair $\type(a)$ of front and rear ranks to each symbol $a\in\AllTerms$. The pair $(\AllTerms,\type)$
is usually identified with $\AllTerms$, and $\type$ is kept implicit.

A string $w=a_1\cdots a_n\in\AllTerms^+$ is \emph{typed} if there are
$k_0,\dots,k_n\in\Nat$ such that
$\type(a_i)=(k_{i-1},k_i)$ for all $i\in[n]$. We let $\type(w)=(k_0,k_n)$. The
set of all typed strings over $\AllTerms$ is written $\AllTerms^\oplus$. The
concatenation $u\cdot v$ of typed strings $u,v\in\AllTerms^\oplus$
with $\type(u)=(i,j)$ and $\type(v)=(m,n)$ is only defined if $j=m$.

Note that ordinary (untyped) alphabets and strings over them can be considered
as special cases of typed ones by setting $\type(a)=(1,1)$ for every symbol~$a$.
Further note that there is no empty typed string $\emptyseq\in\AllTerms^+$ because its
type would be undefined. Instead, we introduce so-called \emph{blanks} below.

For a typed alphabet $\AllTerms$ and a ranked alphabet $\Voc$, an \emph{interpretation
operator} $\sem\cdot \colon \AllTerms \to \GG_\Voc$ assigns a
graph $\sem a\in\GG_\Voc$ with $\type(\sem a)=\type(a)$ to each symbol $a\in \AllTerms$.
We extend $\sem\cdot$ to typed strings over
$\AllTerms$ by $\sem{a_1\cdots a_n}\iso\sem{a_1}\gcomp\cdots\gcomp\sem{a_n}$
where $a_i\in\AllTerms$ for $i\in[n]$.

Every interpretation operator $\sem\cdot \colon \AllTerms \to
\GG_\Voc$ defines a congruence relation $\sim$ on typed strings, as follows:
for all $u,v\in\AllTerms^\oplus$, $u \sim v$ if and only if $\sem u\iso\sem v$.

Given a ranked alphabet $\Voc$, the \emph{canonical alphabet}
$(\AllTerms_\Voc,\type)$ is the typed alphabet given by
\begin{align*}
    \AllTerms_\Voc & = \{a^\phi_\rho\mid \gatom a \phi \rho \text{ is an atom in $\GG_\Voc$}\} \cup \Blanks\\
    \Blanks &=\{\sblank n \rho\mid \gblank n \rho  \text{ is a blank in $\GG_\Voc$}\}\\
    \type(a^\phi_\rho)&=(|\phi|,|\rho|)\quad\text{for all }a^\phi_\rho\in\AllTerms_\Voc\\
    \type(\sblank n \rho)&=(n,|\rho|)\quad\text{for all }\sblank n \rho\in\Blanks.
\end{align*}
The \emph{canonical interpretation} of $\AllTerms_\Voc$ is given by
$\sem{a^\phi_\rho}=\gatom a \phi \rho$ for all $a^\phi_\rho\in\AllTerms_\Voc\setminus\Blanks$ 
and $\sem{\sblank n \rho}=\gblank n \rho$ for all $\sblank n \rho\in\Blanks$.
We call all symbols
in $\Blanks$ \emph{blank symbols} (or just \emph{blanks}), and abbreviate
$\sblank n \rho$ with $\rho = 1\cdots n$ as
$\emptyseq^{(n)}$.
Note that $\Blanks\subseteq\AllTerms_\Voc$ is infinite and independent of $\Voc$.

We note the following immediate consequences of the definitions above:

\vspace{0mm plus5mm}
\pagebreak[2]
\begin{fact}\label{a:powerset-assumptions}
The following holds for every ranked alphabet $\Voc$ and its canonical alphabet $\AllTerms_\Voc$:\nopagebreak
	\begin{enumerate}
		\item \label{as:closure}
		      Let $a,b\in\AllTerms_\Voc$ with $\type(a)=(m,n)$ and $\type(b)=(n,k)$
		      (for some $m,n,k\in\Nat$) be such that $\{a,b\}\cap\Blanks\neq\emptyset$.
		      Then $\AllTerms_\Voc$~contains a symbol $c$ such that $ab\sim c$. In
		      particular, $\type(c)=(m,k)$. (Note that $c\in\Blanks$ if both $a$
		      and $b$ are blanks, and $c\in\AllTerms_\Voc\setminus\Blanks$
		      otherwise.)
		\item \label{as:identity}For all $m,n\in\Nat$ and $a\in\AllTerms_\Voc$ with
		      $\type(a)=(n,m)$, we have $\idblank{n} a \sim a\sim a\,\idblank{m}$.
	\end{enumerate}
\end{fact}

By requiring that finite automata respect types, we obtain finite automata over a
finite typed alphabet $\GVoc$.
\begin{defn}[Finite Automaton]
  Let $\GVoc$ be a finite typed alphabet.
  A \emph{finite automaton} over $\GVoc$ is a tuple
  $\aut A = (\GVoc, Q, \Delta, q_0, F)$ such that $Q$ is a ranked
  alphabet of \emph{states},
  $\Delta \subseteq \left( Q \times \GVoc \times Q \right)$ is a set
  of \emph{transitions} $(q,a,q')$ such that $\type(a)=(i,j)$ implies
  that $\rank(q) = i$ and $\rank(q') = j$, $q_0 \in Q$ is the
  \emph{initial state}, and $F \subseteq Q\setminus\{q_0\}$ is a
  non-empty set of \emph{final states} such that $\rank(q)=\rank(q')$
  for all $q,q'\in F$.
           
  We let $\type(\aut A)=(m,n)$ where $m=\rank(q_0)$ and $n=\rank(q)$
  for all $q\in F$.  $\aut A$ is \emph{deterministic} if, for each
  pair of transitions $(p,a,q), (p',b,q') \in \Delta$, $p=p'$ and
  $a = b$ implies $q = q'$.

  A \emph{configuration} of $\aut A$ is a pair
  $(q,w)\in Q\times\GVoc^\ast$ consisting of the \emph{current state}
  $q$ and the \emph{remaining input} $w$.  It is \emph{initial} if
  $q = q_0$, and \emph{accepting} if $q \in F$ and $w = \emptyseq$.
  Note that configurations are defined for arbitrary $w\in\GVoc^\ast$,
  i.e., $w$ is not necessarily typed.

  A transition $\delta = (q, a, q') \in \Delta$, where $a$ is a symbol
  in $\GVoc$, defines \emph{moves} $(q,a w) \move_\delta (q', w)$ for
  all strings $w \in \GVoc^*$.  We write
  $(q,w) \move_{\aut A} (q', w')$ if $(q,w) \move_\delta (q', w')$ for
  some transition $\delta \in \Delta$.

  The \emph{language} accepted by $\aut A$ is defined as usual:
  \begin{align*}
    \L(\aut A) = \{ w \in \GVoc^*
    \mid \exists q \in F : (q_0,w) \move_{\aut A}^*(q,\emptyseq)\}.
  \end{align*}
\end{defn}
  Note that $\L(\aut A)\subseteq\GVoc^\oplus$ because transitions are typed and $q_0\notin F$.
  In fact, we have $\type(w)=\type(\aut A)$ for all $w\in\L(\aut A)$.
  
Associating a graph interpretation operator with $\GVoc$, we can
  recognize graph languages in the obvious way: a graph $G$ is
  accepted by an automaton $\aut A$ if $G\iso\sem w$ for some string
  $w\in\L(\aut A)$. To enable our automata
  to work directly on graphs, we use \emph{graph configurations} and
  the corresponding moves:

\begin{defn}[Graph configurations and moves]\label{d:graph-conf}%
	Let $\aut A = (\GVoc, Q, \Delta, q_0, F)$ be a finite
	automaton with $\GVoc\subseteq\AllTerms_\Voc$. A pair $(q,G)$ consisting of a
	state $q\in Q$ and a graph $G\in\GG_\Voc$ is a \emph{graph configuration}, or simply
	\emph{configuration} if the context prevents confusion.

	A transition $\delta=(q,a,q')\in\Delta$ can be applied to a
	graph $G\in\GG_\Voc$ if there are graphs $G',G_a\in\GG_\Voc$ such that
	$G = G_a \gcomp G'$ and $G_a\iso\sem a$. We then call $(q,G)
		\gmove_\delta (q',G')$ a \emph{move} (using $\delta$). We write $(q,G)
		\gmove_{\aut A} (q',G')$ if $(q,G) \gmove_\delta (q',G')$ for some
	transition $\delta\in\Delta$. The set of \emph{acceptable
	configurations} of $\aut A$ with $\type(\aut A)=(m,n)$ is 
	$$\CorrectConf{\aut A}=\{(q,G)\in Q\times\GG_\Voc\mid 
	\exists q'\in F,\,G'\in\GG_\Voc\colon (q,G) \gmove_{\aut A}^\ast
		(q',G')\text{ and }G'\iso\gblank n {1\ldots n}\}.$$
	The graph language accepted by $\aut A$ is then
	\[\L_G(\aut
		A)=\{G\in\GG_\Voc\mid (q_0,G)\in\CorrectConf{\aut A}\}.\]
\end{defn}

Obviously (and provable by a straightforward induction), we have
$\L_G(\aut A)=\{\sem w\mid w\in\L(\aut A)\}$. However, given a graph
$G\in\L_G(\aut A)$, it is usually not the case that $w\in\L(\aut A)$
for all $w\in\GVoc^\oplus$ such that $\sem w=G$, as demonstrated in 
the following example. This turns efficient graph
recognition with finite automata into a nontrivial problem.

\begin{figure}[tb]
  \centering
  \begin{automaton}[node distance=6mm]
       \node[state,initial] (q_0) {$q_0$};
        \node[state] (q_1) [right=of q_0] {$q_1$};
        \node[state,accepting](q_2) [right=of q_1] {$q_2$};
        \path[->]
          (q_0) edge[loop above] node {$a^1_1$} ()
                  edge                    node[above] {$a^1_1$} (q_1)
          (q_1) edge                   node[above] {$b^1_2$} (q_2)
    ;
  \end{automaton}
   \caption{The finite automaton $\aut S$ recognizing the graph language of ``stars''.}
   \label{f:star-auto}
\end{figure}
\begin{example}\label{ex:star-auto}
	Let us consider the graph language of ``stars'' where each ``star'' 
	consists of a center node, which is the only front interface node, 
	and at least two satellite nodes, which are connected by binary 
	edges with the center node. Just one of the edges is labeled with~$b$, 
	the others with~$a$, and the rear interface consists of just the 
	satellite node attached to the $b$-labeled edge. 
	\figref{f:atoms-blank} shows such a ``star'', and \figref{f:star-auto}  
	the finite automaton $\aut S$ recognizing this graph 
	language. As usual, we draw finite automata with circles as states  
	and arrows as transitions. The initial state is indicated by an 
	incoming arrow from nowhere, final states by double borders.  
	$\aut S$ is nondeterministic since it contains a transition 
	$(q_0,a^1_1,q_0)$ as well as $(q_0,a^1_1,q_1)$.
	
	One can see that the ``star'' $G$ in \figref{f:atoms-blank} is a member 
	of $\L_G(\aut S)$ because the string $w=a^1_1\,a^1_1\,b^1_2$ 
	is recognized by $\aut S$ and $G\iso\sem w$. Note
	that $\sem w\iso\sem{a^1_1}\gcomp\sem{a^1_1}\gcomp\sem{b^1_2} =
		\gatom{a}{1}{1}\gcomp\gatom{a}{1}{1}\gcomp\gatom{b}{1}{2}$, 
	which is one of the compositions of $G$ shown in 
	\exref{ex:composition}, whereas, e.g.,  the
        string $b^{1}_{12}\,a^{13}_{123}\,a^{134}_{1234}\,\sblank{4}{4}$
	cannot be recognized by $\aut S$.
\end{example}

\section{Efficient Graph Recognition with Finite Automata}
\label{s:pfga}
Let us now use a finite automaton for the recognition of graphs.
To achieve efficiency, we must avoid
nondeterminism. To see this, consider a nondeterministic automaton
and a situation where we have reached a configuration with a state that has
several outgoing transitions reading the same symbol. This means that all of
them are applicable whenever one of them can be applied. Consequently, we must
try one of them first, and if this choice leads into a dead end later, we are
forced to backtrack and then try the next one. Backtracking usually leads to
exponential running times, and should thus be avoided for efficient recognition. 

Note that blank transitions, that is, transitions that read a blank,
can lead to nondeterminism, even if the automaton is in fact
deterministic. To understand this, consider a state with at least two outgoing
transitions, one of them a blank transition. Similarly to an epsilon transition
in an ordinary finite automaton, a blank transition (of the right type) can
always be applied. So whenever one of the other transitions is applicable, one
must choose between that transition and the blank transition. If one picks the
``wrong'' one that leads into a dead end, backtracking is necessary in order to
choose the other one.

For efficient graph recognition with finite automata, we first 
develop an extension of the well-known powerset construction to
turn a finite automaton over a subset of the canonical alphabet
into a deterministic one without blank transitions in places where they may
trigger backtracking.

\subsection{Powerset Construction for Finite Automata\label{sec:powerset automaton}}

\paragraph*{General Assumption.} Throughout this
subsection, we consider a fixed (ranked) alphabet $\Voc$,
its canonical alphabet $\AllTerms_\Voc$, and a finite automaton 
$\aut A = (\GVoc, Q, \Delta, q_0, F)$ over a typed alphabet
$\GVoc\subset\AllTerms_\Voc$. We also let $(m,n)=\type(\aut A)$.

\begin{algorithm}[tb]
	\caption{Powerset construction for a finite automaton.}
	\label{alg:powerset}
	\Input{Finite automaton $\aut A = (\GVoc, Q, \Delta, q_0, F)$ with $\GVoc\subset\AllTerms_\Voc$.}
	\Output{Powerset automaton $\aut A' = (\GVoc', Q', \Delta', S_0, F')$.}
	$S_0 \leftarrow \closure{q_0}$ and $S_f\leftarrow\emptyset$\;
	$Q' \leftarrow  \{S_0,S_f\}$ and $F'\leftarrow\{S_f\}$ where $(\rank(S_0),\rank(S_f))=\type(\aut A)$\;\label{a:initial}
	$\Delta' \leftarrow \emptyset$\;
	$\VWork \leftarrow \{S_0\}$\;\label{a:initial-W}
	\While{$\VWork \neq \emptyset$\label{a:while-loop}}{
		select and remove any $X$ from $\VWork$\;\label{a:select-Q}
		\lForEach{$(\beta,q)\in X$ such that $q\in F$}{
			add $(X,\beta,S_f)$ to $\Delta'$\label{a:add-final-prod}
		}
		$\psi \leftarrow \{(a, q') \in (\AllTerms\setminus\Blanks)\times Q\mid\text{$(\beta,q) \in X$, $(q,b,q')\in\Delta$, and $a\sim \beta b$}\}$\;\label{a:compute-M}
		\ForEach{$a_0\in\domain\psi$\label{a:select-a}}{
			let $\type(a_0)=(i,j)$\;
			$Y \leftarrow \closure{\psi(a_0)}$\;\label{a:compute-Q'}
			\If{$Y=\beta'Y'$ for some $Y'\in Q'$ and $\beta'\in\Blanks$ with $\type(\beta')=(j,j)$\label{a:Q'-case-2}}%
			{
				add $(X,a_1,Y')$ to $\Delta'$ for some $a_1\in\Terms\setminus\Blanks$ such that $a_1\sim a_0\beta'$\;\label{a:add-transition-1}
			}
			\Else{\label{a:Q'-case-3}
				let $\rank(Y)=j$ and add $Y$ to $Q'$ as well as to $\VWork$\;\label{a:expand-nonterms}
				add $(X,a_0,Y)$ to $\Delta'$\;\label{a:add-transition-2}
			}
		}
	}
	$\GVoc'\leftarrow\{a\in\AllTerms_\Voc\mid\exists S,S''\in Q':(S,a,S')\in\Delta'\}$
\end{algorithm}
\algref{alg:powerset} extends the well-known powerset construction and computes a new 
automaton whose states are built from sets of states of the input automaton. As mentioned
above, blank 
transitions resemble epsilon transitions that ``read'' the empty word. 
However, blanks are more complicated because they must take 
nodes into account; remember that the rear interface of a blank can contain a subset 
of its nodes in any order (cf.~\figref{f:atoms-blank}). Therefore, 
\algref{alg:powerset} considers pairs of the form $(\beta,q)$ where $\beta$ is a 
blank and $q$ a state of the input automaton. Informally speaking, such a pair 
indicates that one can reach state $q$ after reading $\beta$. Note that $\beta$ 
can be $\idblank{i}$ for an appropriate $i$, that is, an identity without any effect 
(which is comparable to $\emptyseq$ in the string case). The states of the computed 
automaton then consist of sets of such pairs. 

When building these sets, the algorithm must follow any sequence of blank transitions 
and combine their blanks into a single blank, which is always possible thanks to 
\factref{a:powerset-assumptions}. We call the set of all states (together with 
these blanks) that can be reached from a state $q\in Q$ by a sequence of blank transitions
the \emph{closure} of $q$. It is defined as
$$
	\closure{q}=\{(\beta,q')\in\Blanks\times Q\mid
	\exists \sigma\in\Blanks^\ast\colon
	(q,\sigma)\move_{\aut A}^\ast(q',\emptyseq)\text{ and }\idblank{i}\beta\sim\idblank{i}\sigma\}
$$
where $\rank(q)=i$. The condition $\idblank{i}\beta\sim\idblank{i}\sigma$ in this definition can be
simplified to $\beta\sim\sigma$ unless $\sigma=\emptyseq$. If $\sigma=\emptyseq$
then $\idblank{i}$ is required because $\sem\emptyseq$ is undefined. Note also
that $(\idblank{i},q)\in\closure q$, and that $(\beta,q')\in\closure q$ implies
$\type(\beta)=(i,\rank(q'))$.

We extend the definition of closures 
to sets of states by defining $\closure M=\bigcup_{q\in M}\closure q$ for all
$M\subseteq Q$ as used in \alineref{a:compute-Q'}.
For every set $C\subseteq\Blanks\times Q$ and $i\in\Nat$ such that $(\beta',q)\in
	C$ implies $\type(\beta')=(i,\rank(q))$, and every blank $\beta\in\Blanks$ of
type $(i,i)$, we define
\[
	\beta C = \{(\beta'',q)\in\Blanks\times Q\mid\text{$\exists\beta'\in\Blanks\colon(\beta',q)\in C$ and $\beta''\sim \beta\beta'$}\}
\]
as used in \alineref{a:Q'-case-2}. Note that a blank
$\beta''$ with $\beta''\sim \beta\beta'$ always exists, thanks to
\factsubref{a:powerset-assumptions}{as:closure}.

We now prove the correctness of \algref{alg:powerset} by first showing that it
always terminates and produces a proper finite automaton as a result. After
that, we will show that the produced powerset automaton is equivalent to the
input automaton. 

\begin{lemma}
	\algref{alg:powerset} terminates for every finite automaton $\aut A =
		(\GVoc, Q, \Delta, q_0, F)$ and outputs a finite
	automaton.
\end{lemma}

\begin{proof}
	We use the notations established in \algref{alg:powerset}. In particular,
	$\aut A' = (\GVoc', Q', \Delta', S_0, F')$ is the output
        of the  algorithm.
	
	Every state in $Q'$ is a subset of $\Blanks\times Q$. The fact that $Q$ and
	$\Terms$ are finite thus implies that $Q'$ is finite. Since no element of
	$Q'$ is added to $\VWork$ more than once, none of the loops of
	\algref{alg:powerset} can run endlessly.

	It remains to be shown that the transitions of $\aut A'$ are correctly typed.
	For this, we consider any transition
	$\delta=(S,a,S')\in\Delta'$ and show that $\type(a) =
        (\rank(S),\rank(S'))$.
	Transition $\delta$ must have been added in \alinesref{a:add-final-prod},
	\ref{a:add-transition-1}, or~\ref{a:add-transition-2}. Before we consider
	each of these cases, first note that each state $S\neq S_f$ is a non-empty
	set of pairs $(\beta,q)\in\Blanks\times Q$ that satisfy
	$\type(\beta)=(\rank(S),\rank(q))$ by the definition of closures, how $\psi$
	is computed in \alineref{a:compute-M}, and how $\rank(S)$ is set in
	\alinesref{a:initial} and~\ref{a:expand-nonterms}. Further note that we have
	$\rank(X)=i$ when \alineref{a:Q'-case-2} is reached, also by the definition
	of closures and how $\psi$ is computed. We now distinguish the three cases
	where $\delta$ may have been added to $\Delta'$:
	\begin{itemize}
		\item \alineref{a:add-final-prod}: Since $(\beta,q)\in X$ and $q\in
			      F$, $\type(\beta)=(\rank(X),\rank(q))=(\rank(X),\rank(S_f))$
		\item \alineref{a:add-transition-1}: $\rank(Y')=j$ and
		      $\type(a_1)=(i,j)$ follow from $Y=\beta'Y'$, $\type(\beta')=(j,j)$, and
		      $\type(a_0)=(i,j)$, and hence $\type(a_1)=(\rank(X),\rank(Y'))$.
		\item \alineref{a:add-transition-2}: $\rank(Y)$ is set to $j$, and hence
		      $\type(a_0)=(\rank(X),\rank(Y))$,
	\end{itemize}
	Hence, the transitions in $\Delta'$ are correctly typed, i.e., $\aut A'$ is a
	finite automaton.
\end{proof}

To prove that the powerset automation $\aut A'$ produced by
\algref{alg:powerset} is equivalent to its input automaton~$\aut A$, we now show
that every accepting sequence of moves in $\aut A$ has a corresponding sequence
of moves in $\aut A'$ (\lemmaref{lemma:Gamma-to-Gamma'}) and vice versa
(\lemmaref{lemma:Gamma'-to-Gamma}). For proving
\lemmaref{lemma:Gamma-to-Gamma'}, we will need the following auxiliary result.
Informally speaking (and ignoring the fact that we are actually dealing with
pairs $(\beta,q)$ instead of just states $q$), \lemmaref{l:closure} proves that
every state reachable from a closure state  by a blank transition is also
contained in the closure: 

\begin{lemma}\label{l:closure}%
	Let $q,q',q''\in Q$, $\beta,\beta',\beta''\in\Blanks$, and
	$\delta=(q',\beta',q'')\in\Delta$. If $(\beta,q')\in\closure q$ and
	$\beta''\sim \beta\beta'$ then $(\beta'',q'')\in\closure q$.
\end{lemma}

\begin{proof}
	Let $\rank(q)=i$. By the definition of $\closure q$, $(\beta,q')\in\closure
		q$ implies $(q,\sigma)\move_{\aut A}^\ast(q',\emptyseq)$ and
	$\idblank{i}\beta\sim\idblank{i}\sigma$ for some $\sigma\in\Blanks^\ast$.
	Consequently, $(q,\sigma\beta')\move_{\aut
			A}^\ast(q',\beta')\move_\delta(q'',\emptyseq)$ and $\idblank{i}\sigma\beta'
		\sim \idblank{i}\beta\beta' \sim \idblank{i}\beta''$, and thus
	$(\beta'',q'')\in\closure q$.
\end{proof}

\begin{lemma}\label{lemma:Gamma-to-Gamma'}%
	Consider $q\in Q$ and $w\in\Terms^\ast$ such that $(q_0,w)\move_{\aut
			A}^\ast(q,\emptyseq)$. Then there exist $w'\in\Terms^\ast$, $S\in Q'$, and
	$\beta \in\Blanks$ such that $(S_0,w')\move_{\aut A'}^\ast(S,\emptyseq)$,
	$(\beta,q)\in S$, and $\idblank{m}w\sim \idblank{m}w'\beta$.\footnote{Recall that $\type(\aut A)=(m,n)$.}
\end{lemma}

\begin{proof}
	We prove the proposition by induction over the length $\ell$ of the move sequence
	$(q_0,w)\move_{\aut A}^\ast(q,\emptyseq)$. 
	The proposition is clearly true for $\ell=0$ (with $w=w'=\emptyseq$, $S=S_0$, and
	$\beta=\idblank{m}$). For $\ell>0$, we consider the sequence
	$(q_0,wa)\move_{\aut A}^{\ell-1}(q,a)\move_\delta(q',\emptyseq)$ using
	transition $\delta=(q,a,q')\in\Delta$. By the induction hypothesis,
	there is a sequence $(S_0,w')\move_{\aut A'}^\ast(S,\emptyseq)$ with
	$(\beta,q)\in S$ and $\idblank{m}w\sim \idblank{m}w'\beta$. We have to
	show that there is also a sequence $(S_0,w'a')\move_{\aut
			A'}^\ast(S',\emptyseq)$ with $(\gamma,q')\in S'$ and $\idblank{m}wa\sim
		\idblank{m}w'a'\gamma$. We distinguish two cases:
	\begin{itemize}
		\item $a\in\Blanks$: By \alinesref{a:initial}, \ref{a:compute-Q'}, and
		      \ref{a:expand-nonterms} of \algref{alg:powerset},
		      $(\beta,q)\in\closure{q''}\subseteq S$ for some $q''\in Q$. Let
		      $\beta'\in\Blanks$ such that $\beta'\sim \beta a$. We have
		      $(\beta',q')\in\closure{q''}\subseteq S$ by \lemmaref{l:closure}
		      and $\idblank{m}wa\sim \idblank{m}w'\beta a \sim
			      \idblank{m}w'\beta'$ by the induction hypothesis and the
		      definition of $\beta'$. This proves the proposition by choosing
		      $a'=\emptyseq$ and $\gamma =\beta'$.
		\item $a\in\Terms\setminus\Blanks$: Whenever an element is added to
		      $Q'$, it is also added to $\VWork$ (\alineref{a:initial-W}
		      and~\alineref{a:expand-nonterms}). Since $S\in Q'$ and
		      $\VWork=\emptyset$ when \algref{alg:powerset} terminates, $S$ must have
		      been selected as $X$ in \alineref{a:select-Q} at some point. As
		      $(\beta, q)\in S$ and $\delta=(q,a,q')\in\Delta$, there is a
		      $b\in\Terms\setminus\Blanks$ such that $b\sim
			      \beta a$, $(b,q')\in \psi$ after executing \alineref{a:compute-M},
		      and $(\idblank{j},q')\in Y$ after executing
		      \alineref{a:compute-Q'}, and after selecting $b$ as $a_0$ in
		      \alineref{a:select-a}. $Y$ is handled in one of two sub-cases
		      selected in \alinesref{a:Q'-case-2} and~\ref{a:Q'-case-3}:
		      \begin{itemize}
			      \item \alineref{a:Q'-case-2}, i.e., $Q'$ already contains a
			            similar nonterminal $Y'$ such that $Y=\beta'Y'$ for some
			            blank $\beta'\in\Blanks$ with $\type(\beta')=(j,j)$, and
			            $\delta'=(X,b',Y')=(S,b',Y')$ is added to~$\Delta'$
			            where $b'\in\Terms\setminus\Blanks$ and $b'\sim
				            b\beta'=a_0\beta'$. Consequently, $(S_0,w'b')\move_{\aut
					            A'}^\ast(S,b')\move_{\delta'}(Y',\emptyseq)$. Since
			            $(\idblank{j},q')\in Y$, there exists $\beta''\in\Blanks$
			            such that $(\beta'',q')\in Y'$ and $\idblank{j}\sim
				            \beta'\beta''$. Therefore,
			            \[
				            \idblank{m}wa \sim \idblank{m}w'\beta a \sim
				            \idblank{m}w'b \sim \idblank{m}w'b\idblank{j}
				            \sim\idblank{m}w'b\beta'\beta'' \sim
				            \idblank{m}w'b'\beta''
			            \]
			            by the induction hypothesis, the construction of $b$ and
			            $b'$, and the fact that $\idblank{j}\sim \beta'\beta''$.
			            This proves the proposition by choosing $a'=b'$, $\gamma
				            =\beta''$, and $S'=Y'$.
			      \item \alineref{a:Q'-case-3}, i.e., $Y$ is not already in
			            $Q'$, and neither is there a similar one. Thus, $Y$ is
			            added to $Q'$ and $\delta'=(X,b,Y)=(S,b,Y)$ is added to
			            $\Delta'$. Consequently, $(S_0,w'b)\move_{\aut
					            A}^\ast(S,b)\move_{\delta'}(Y,\emptyseq)$ and
			            $\idblank{m}wa \sim \idblank{m}w'\beta a \sim
				            \idblank{m}w'b$ by the induction hypothesis and the
			            construction of $b$. This proves the proposition with
			            $a'=b$, $\gamma =\emptyseq$, and
			            $S'=Y$.\qedhere
		      \end{itemize}
	\end{itemize}
\end{proof}

\begin{lemma}\label{lemma:Gamma'-to-Gamma}
	Let $S\in Q'$, $\beta \in\Blanks$, $q\in Q$, and $w'\in\Terms^\ast$ be such
	that $(\beta,q)\in S$ and $(S_0,w')\move_{\aut A'}^\ast(S,\emptyseq)$. Then
	there exists $w\in\Terms^\ast$ such that $(q_0,w)\move_{\aut
			A}^\ast(q,\emptyseq)$ and $\idblank{m}w \sim \idblank{m}w'\beta$.
\end{lemma}

\begin{proof}
	We prove the proposition by induction over the length $\ell$ of the
	move sequence $(S_0,w')\move_{\aut A'}^\ast(S,\emptyseq)$.
	For $\ell=0$, we have $(S_0,\emptyseq)\move_{\aut A'}^0(S_0,\emptyseq)$.
	Consider any $(\beta,q)\in S_0$. $S_0=\closure{q_0}$ implies
	$(q_0,\sigma)\move_{\aut A}^\ast(q,\emptyseq)$ for some
	$\sigma\in\Blanks^\ast$ such that $\idblank{m}\sigma \sim \idblank{m}\beta$.
	This proves the proposition by choosing $w'=\emptyseq$ and $w=\sigma$.

	For $\ell>0$, the sequence $(S_0,w')\move_{\aut A'}^\ast(S,\emptyseq)$ has the
	form $(S_0,w_0'a)\move_{\aut A'}^{\ell-1}(S',a)\move_\delta(S,\emptyseq)$,
	where $w'=w_0'a$ and $\delta=(S',a,S)\in\Delta'$. By the induction
	hypothesis, for every $(\gamma,q')\in S'$, there is $w_0\in\Terms^\ast$ such
	that $(q_0,w_0)\move_{\aut A}^\ast(q',\emptyseq)$ and $\idblank{m}w_0 \sim
		\idblank{m}w_0'\gamma$. We show that $(q_0,w)\move_{\aut
			A}^\ast(q,\emptyseq)$ for some $w\in\Terms^\ast$ such that $\idblank{m}w
		\sim \idblank{m}w_0'a\beta$. This will complete the proof because
	$w'=w_0'a$.

	Transition $\delta=(S',a,S)$ must have been added to $\Delta'$ in one of the
	two cases selected in \alinesref{a:Q'-case-2} or \ref{a:Q'-case-3}.
	\begin{itemize}
		\item \alineref{a:Q'-case-2}: $\delta=(S',a,S)$ has been added to
		      $\Delta'$ after $S'$ had been selected as $X$ in
		      \alineref{a:select-Q}, a terminal $a_0\in\Terms\setminus\Blanks$
		      had been selected in \alineref{a:select-a}, $Y$ had been computed
		      in \alineref{a:compute-Q'}, and $S=Y'$ had been identified by
		      $Y=\beta'Y'$ for some blank $\beta'\in\Blanks$ with
		      $\type(\beta')=(j,j)$ and $a \sim a_0\beta'$. Since $(\beta,q)\in
			      S=Y'$, $Y=\beta'Y'$ must contain some $(\gamma_0,q)$ with
		      $\gamma_0 \sim \beta'\beta$. There must be a state $q''\in Q$ such
		      that $(\gamma_0,q)\in\closure{q''}$ and $q''\in\psi(a_0)$, where
		      $\psi$ is computed in \alineref{a:compute-M}, because
		      $(\gamma_0,q)\in Y$. By the definition of $\closure X$,
		      $(\gamma_0,q)\in\closure{q''}$ implies $(q'',\sigma)\move_{\aut
				      A}^\ast(q,\emptyseq)$ for some $\sigma\in\Blanks^\ast$ such that
		      $\idblank{j}\sigma \sim \idblank{j}\gamma_0$. By the definition of
		      $\psi$, there must be a transition $(q',b,q'')\in\Delta$ and
		      $(\gamma,q')\in S'=X$ such that $a_0 \sim \gamma b$ and thus
		      $(q_0,w_0b\sigma)\move_{\aut A'}^\ast(q',b\sigma)\move_{\aut
				      A}(q'',\sigma)\move_{\aut A}^\ast(q,\emptyseq)$ by the induction
		      hypothesis for some $w_0\in\Terms^\ast$ such that
		      $\idblank{m}w_0'\gamma \sim \idblank{m}w_0$. Therefore,
		      \[
			      \idblank{m}w_0b\sigma \sim
			      \idblank{m}w_0'\gamma b\sigma \sim
			      \idblank{m}w_0'a_0\sigma \sim
			      \idblank{m}w_0'a_0\gamma_0 \sim
			      \idblank{m}w_0'a_0\beta' \beta \sim
			      \idblank{m}w_0'a\beta.
		      \]
		      This proves the statement of the lemma in this case by choosing
		      $w=w_0b\sigma$.
		\item \alineref{a:Q'-case-3}: $\delta=(S',a,S)$ has been added to
		      $\Delta'$ after $S'$ had been selected as $X$ in
		      \alineref{a:select-Q} and $S=Y$ had been computed in
		      \alineref{a:compute-Q'} after selecting $a$ as $a_0$ in
		      \alineref{a:select-a}. There must be a state $q''\in Q$ such that
		      $(\beta,q)\in\closure{q''}$ and $(a,q'')\in\psi$ computed in
		      \alineref{a:compute-M} because $(\beta,q)\in S=Y$. Similar to the
		      previous case, $(\beta,q)\in\closure{q''}$ implies
		      $(q'',\sigma)\move_{\aut A}^\ast(q,\emptyseq)$ for some
		      $\sigma\in\Blanks^\ast$ such that $\idblank{j}\sigma \sim
			      \idblank{j}\beta$. By the definition of $\psi$, there must be a
		      transition $(q',b,q'')\in\Delta$ and $(\gamma,q')\in S'=X$ with $a
			      \sim \gamma b$ and thus $(q_0,w_0b\sigma)\move_{\aut
				      A}^\ast(q',b\sigma)\move_{\aut A}(q'',\sigma)\move_{\aut
				      A}^\ast(q,\emptyseq)$ by the induction hypothesis for some
		      $w_0\in\Terms^\ast$ such that $\idblank{m}w_0'\gamma \sim
			      \idblank{m}w_0$. Therefore,
		      \[
			      \idblank{m}w_0b\sigma \sim
			      \idblank{m}w_0'\gamma b\sigma \sim
			      \idblank{m}w_0'a\sigma \sim
			      \idblank{m}w_0'a\beta.
		      \]
		      This proves the statement in this case, and thus completes the
		      proof of the lemma, by setting $w=w_0b\sigma$.\qedhere
	\end{itemize}
\end{proof}

\noindent We are now ready to prove the equivalence of a 
finite automaton and its powerset automaton using the following notion of equivalence:

\begin{defn}[Equivalent automata]\label{def:equivalent-automata}%
	Two finite automata of the same type $(m,n)$ are \emph{equivalent} if the following holds: for
	every $w\in\L(\aut A)$ there exists $w'\in\L(\aut A')$ such that
	$\idblank{m}w \sim \idblank{m}w'$, and for every $w'\in\L(\aut A')$ there
	exists $w\in\L(\aut A)$ such that $\idblank{m}w \sim \idblank{m}w'$.
\end{defn}

\begin{theorem}
	$\aut A$ and the powerset automaton $\aut A'$ computed from $\aut A$
	by \algref{alg:powerset} are equivalent.
\end{theorem}

\begin{proof}
	We have to prove both implications in \defref{def:equivalent-automata}. To
	prove the first implication, let $w\in\Terms^\ast$ and consider a sequence
	$(q_0,w)\move_{\aut A}^\ast(q,\emptyseq)$ with $q\in F$.
	By \lemmaref{lemma:Gamma-to-Gamma'}, there exist
	$w''\in\Terms^\ast$, $S\in Q'$, and $\beta \in\Blanks$ such that
	$(S_0,w'')\move_{\aut A'}^\ast(S,\emptyseq)$, $(\beta,q)\in S$, and
	$\idblank{m}w\sim \idblank{m}w''\beta$. Since $S\in Q'$, $S$ must have
	been selected as $X$ in \alineref{a:select-Q} at some point. Since $q\in
		F$ and $(\beta,q)\in S$, a transition $(S,\beta,S_f)$ has been added to
	$\Delta'$ in \alineref{a:add-final-prod}, and thus
	$(S_0,w''\beta)\move_{\aut A'}^\ast(S,\beta)\move_{\aut A'}(S_f,\emptyseq)$,
	i.e., $w''\beta\in\L(\aut A')$, which proves the first implication by
	choosing $w'=w''\beta$.

	Now consider any $w'\in\Terms^*$ and a sequence $(S_0,w')\move_{\aut
			A'}^\ast(S_f,\emptyseq)$. Since $S_f$ can only be reached by
	transitions added in \alineref{a:add-final-prod}, this sequence must
	have the form $(S_0,w''\beta)\move_{\aut
			A'}^\ast(S,\beta)\move_{\aut A'}(S_f,\emptyseq)$ for some $S\in q'$,
	$w''\in\Terms^\ast$, $q\in F$, and $\beta\in\Blanks$ such that
	$w'=w''\beta$ and $(\beta,q)\in S$. By
	\lemmaref{lemma:Gamma'-to-Gamma}, there is a $w\in\Terms^\ast$ such
	that $(q_0,w)\move_{\aut A}^\ast(q,\emptyseq)$ and $\idblank{m}w
		\sim \idblank{m}w''\beta=\idblank{m}w'$, which proves the other
	implication.
\end{proof}

\subsection{Making the Powerset Automaton Deterministic\label{sec:deterministic}}

Since \algref{alg:powerset} resembles the classical powerset construction, one might
assume that the resulting powerset automaton is
deterministic, but this is not always the case, demonstrated by the following example:

\begin{example}\label{ex:ambiguous-auto}%
  \begin{figure}[tb]
    \begin{minipage}{0.45\textwidth}
      \centering
      \begin{automaton}[node distance=12mm]
        \node[state,initial]       (q_0) {$q_0$};
        \node[state,accepting] (q_1) [right=of q_0] {$q_1$};
        \node[state]                 (q_2) [below=of q_1] {$q_2$}; 
        \node[state,accepting] (q_3) [below=of q_0] {$q_3$}; 
        \path[->]
          (q_0) edge node[above] {$a^{12}_{12}$} (q_1)
                  edge node[below,sloped] {$b^{12}_{21}$} (q_2)
                  edge node[left] {$b^{12}_{12}$} (q_3)
          (q_1) edge node[left] {$\emptyseq^{(2)}_{21}$} (q_2) 
          (q_2) edge[bend right] node[right] {$\emptyseq^{(2)}_{21}$} (q_1) 
        ;
      \end{automaton} 
      \caption{Automaton $\aut B$ of \exref{ex:ambiguous-auto}.}
      \label{f:ambiguous-auto}
    \end{minipage}
    \hfill
    \begin{minipage}{0.45\textwidth}
      \centering
      \begin{automaton}
        \node[state,initial]       (S_0) {$S_0$};
        \node[state] (S_1) [right=of S_0] {$S_1$};
        \node[state]                 (S_2) [below=of S_0] {$S_2$}; 
        \node[state,accepting] (S_f) [below=of S_1] {$S_f$}; 
        \path[->]
          (S_0) edge[bend left] node[above] {$a^{12}_{12}$} (S_1)
                  edge[bend right]                 node[below] {$b^{12}_{12}$} (S_1)
                  edge                 node[left]     {$b^{12}_{12}$} (S_2)
          (S_1) edge                node[right]   {$\emptyseq^{(2)}$} (S_f) 
          (S_2) edge node[below] {$\emptyseq^{(2)}$} (S_f) 
        ;
      \end{automaton} 
      \caption{Nondeterministic automaton $\aut B'$ produced from $\aut B$ by \algref{alg:powerset}.}
      \label{f:ambiguous-auto-det}
    \end{minipage}
  \end{figure}
	Let $a$ and $b$ be symbols of rank 2, and consider the finite automaton $\aut B$ shown in \figref{f:ambiguous-auto} as 
	input to \algref{alg:powerset}.
	Before the first iteration of the while loop beginning at
	\alineref{a:while-loop} starts, we have
	$S_0=\{(\sblank{2}{12},q_0)\}$, $Q'=\{S_0,S_f\}$,
	$\Delta'=\emptyset$, and $\VWork=\{S_0\}$. The first
	iteration of the while loop then selects $X=S_0$ and computes
	$\psi=\{(a^{12}_{12},q_1),(b^{12}_{21},q_2),(b^{12}_{12},q_3)\}$. The nested
	foreach loop has the following iterations:

	\begin{enumerate}
		\item $a_0=a^{12}_{12}$ and $Y=\{(\sblank{2}{12},q_1),
			      (\sblank{2}{21},q_2)\}$.
		      \alinesref{a:expand-nonterms} and~\ref{a:add-transition-2} add $Y$
		      as a new state $S_1$ to $Q'$ and $(S_0,a^{12}_{12}, S_1)$ to
		      $\Delta'$.

		\item  $a_0=b^{12}_{21}$ and $Y=\{(\sblank{2}{12},q_2),
			      (\sblank{2}{21},q_1)\}$.
		      Since $Y=\sblank{2}{21}S_1$, \alineref{a:add-transition-1}
		      computes $a_1=b^{12}_{12}$ because $b^{12}_{12} \sim
			      b^{12}_{21}\,\sblank{2}{21}$, and adds
		      $(S_0,b^{12}_{12},S_1)$ to $\Delta'$.

		\item  $a_0=b^{12}_{12}$ and $Y=\{(\sblank{2}{12},q_3)\}$. 
		      \alinesref{a:expand-nonterms}
		      and~\ref{a:add-transition-2} add $Y$ as a new state $S_2$ to $Q'$
		      and $(S_0,b^{12}_{12},S_2)$ to $\Delta'$.
	\end{enumerate}
	\algref{alg:powerset} terminates after adding
	$(S_1,\sblank{2}{12},S_f)$ and
	$(S_2,\sblank{2}{12},S_f)$ to $\Delta'$ with 
	$Q'=\{S_0,S_1,S_2,\allowbreak S_f\}$. 
	\figref{f:ambiguous-auto-det} shows the resulting automaton $\aut B'$.
	We have $(S_0,b^{12}_{12},S_1),
		(S_0,b^{12}_{12},S_2)\in\Delta'$, that is, $\aut B'$ is
	nondeterministic.\qed
\end{example}

The reason for the nondeterminism of the resulting automaton is apparently the
use of the atoms $b^{12}_{12}$ and $b^{12}_{21}$, which is problematic because $b^{12}_{12} \sim
	b^{12}_{21}\,\sblank{2}{21}$. Let us call automata that use such
symbols \emph{ambiguous}:

\begin{defn}[Ambiguous finite automata]\label{d:ambiguous}%
	A finite automaton
	$\aut A = (\GVoc, Q, \Delta, q_0, F)$ is
	\emph{ambiguous} if there are symbols $a,a'\in\GVoc\setminus\Blanks$ such
	that $a\neq a'$ and $a\beta \sim a'\beta'$ for some
	$\beta,\beta'\in\Blanks$
	(i.e., $\sem a$ and $\sem{a'}$ differ only in their rears).
    $\aut A$ is called \emph{unambiguous} if it is not ambiguous.
\end{defn}

Note that the automaton defined in \exref{ex:ambiguous-auto} is ambiguous by
this definition because $b^{12}_{21}\,\sblank{2}{21} \sim b^{12}_{12} \sim
	b^{12}_{12}\,\sblank{2}{12}$.

\begin{theorem}\label{thm:deterministic}
	The powerset automaton $\aut{A'}$ computed by \algref{alg:powerset} is deterministic if $\aut A$ is unambiguous.
\end{theorem}

\begin{proof}
	Consider any unambiguous automaton $\aut A$, and assume that
	\algref{alg:powerset} produces a nondeterministic automaton $\aut A'$, i.e.,
	there are two transitions $\delta=(S,a,S')\in\Delta'$ and
	$\delta'=(S,a,S'')\in\Delta'$ such that $S' \neq S''$. Note
	that we have target state $S_f$ and $\beta\in\Blanks$ for all transitions
	added to $\Delta'$ in \alineref{a:add-final-prod}, and $Y\neq S_f$, $Y'\neq
		S_f$, $a_0\notin\Blanks$, as well as $a_1\notin\Blanks$ for all transitions
	added in \alinesref{a:add-transition-1} and~\ref{a:add-transition-2},
	respectively. Consequently, neither $\delta$ nor $\delta'$ can have been
	added by \alineref{a:add-final-prod}, and both of them must have been added
	within the same iteration of the while loop using the same set $\psi$
	computed in \alineref{a:compute-M}, but in different iterations of the
	foreach loop that have selected, say, $c_1$ and $c_2$ as $a_0$,
	respectively. If $\delta$ has been added to $\Delta'$ in
	\alineref{a:add-transition-1}, we have $a \sim c_1\beta_1$ for some
	$\beta_1\in\Blanks$, and $a=c_1$ if $\delta$ has been added in
	\alineref{a:add-transition-2}. Similarly, $a\sim c_2\beta_2$ for some
	$\beta_1\in\Blanks$, or $a=c_2$. We have $c_1 \neq c_2$ because $c_1$ and
	$c_2$ have been selected in different foreach loop iterations.
	Hence, either $c_1 \sim c_2\beta_2$, $c_1\beta_1\sim c_2$, or
	$c_1\beta_1\sim c_2\beta_2$, in contradiction to $\aut A$ being unambiguous.
\end{proof}

\begin{example}\label{ex:deterministic-auto}%
  \begin{figure}[tb]
  \centering
  \begin{automaton}
    \node[state,initial] (S_0) {$S_0$};
    \node[state] (S_1) [right=of S_0] {$S_1$};
    \node[state](S_2) [right=of S_1] {$S_2$};
    \node[state,accepting](S_f) [right=of S_2] {$S_f$};
    \path[->]
        (S_0) edge node[above] {$a^1_1$}
                        node[below,text=\Red] {$\delta_1$} (S_1)
        (S_1) edge node[above] {$b^1_2$}
                        node[below,text=\Red] {$\delta_3$} (S_2)
                edge[in=60,out=120,looseness=8,loop]
                        node[above] {$a^1_1$}
                        node[below,text=\Red] {$\delta_2$} ()
        (S_2) edge node[above] {$\emptyseq^{(1)}$}
                        node[below,text=\Red] {$\delta_4$} (S_f)
    ;
  \end{automaton}
  \caption{The deterministic automaton $\aut S'$ obtained from $\aut S$ (\figref{f:star-auto}) by \algref{alg:powerset}.}
  \label{f:det-star-auto}
  \end{figure}
	The automaton $\aut S$ shown in \figref{f:star-auto} is unambiguous, and 
	\algref{alg:powerset} thus computes a deterministic finite automaton (DFA) $\aut S'$ shown in \figref{f:det-star-auto}.
	The states of $\aut S'$ are $S_0=\{(\idblank{1},q_0)\}$, 
	$S_1=\{(\idblank{1},q_0),(\idblank{1},q_1)\}$, $S_2=\{\idblank{1},q_2)\}$, 
	and $S_f=\emptyset$.
\end{example}

The restriction to unambiguous automata is in fact insignificant for automata over
subsets of the canonical alphabet $\AllTerms_\Voc$ (and its canonical interpretation).
One can easily transform an automaton $\aut A$ into an equivalent unambiguous one using the following
iterative process:

Consider an automaton $\aut A$ that has $a,a'\in\GVoc\setminus\Blanks$ such that
$a \ne a'$ and $a\beta \sim a'\beta'$ for some $\beta,\beta'\in\Blanks$. Since
we use the canonical interpretation of symbols in $\AllTerms_\Voc$, $a$ and $a'$
must have the form $a=b^\phi_\rho$ and $a'=b^\phi_{\rho'}$ using the same label
$b\in\ed\Voc$ and front interface $\phi$, that is, $a$ and $a'$ differ only in
their rear interfaces. Hence, one can identify a more general rear interface
$\hat\rho$ such that $a \sim b^\phi_{\hat\rho}\,\hat\beta$ and $a' \sim
	b^\phi_{\hat\rho}\,\hat\beta'$ for some blanks $\hat\beta$ and $\hat\beta'$.
$\aut A'$ is then obtained from $\aut A$ by removing $a$ and $a'$ from $\GVoc$
and adding $b^\phi_{\hat\rho}, \hat\beta$, and $\hat\beta'$. Furthermore,
every transition $(q,a,q')$ in $\Delta$ is replaced by two transitions
$(q,b^\phi_{\hat\rho},q'')$ and $(q'',\hat\beta,q')$ where $q''$ is a new state
with $\rank(q'')=|\hat\rho|$ that is also added to $Q$, and proceeds similarly
for each transition referring to $a'$. This process is continued until $\aut A'$
is finally unambiguous.

\begin{example}
  \begin{figure}[tb]
    \begin{minipage}{0.5\textwidth}
     \centering
      \begin{automaton}
        \node[state,initial]       (q_0) {$q_0$};
        \node[state]                 (q_3p) [below=of q_0] {$q_3'$}; 
        \node[state,accepting] (q_3) [left=of q_3p] {$q_3$}; 
        \node[state]                 (q_2p) [right=of q_3p] {$q_2'$}; 
        \node[state]                 (q_2) [right=of q_2p] {$q_2$}; 
        \node[state,accepting] (q_1) [above=of q_2] {$q_1$};
        \path[->]
          (q_0) edge node[above] {$a^{12}_{12}$} (q_1)
                  edge node[right] {$b^{12}_{12}$} (q_2p)
                   edge node[left] {$b^{12}_{12}$} (q_3p)
           (q_1) edge[bend right] node[left] {$\emptyseq^{(2)}_{21}$} (q_2) 
          (q_2p) edge node[below] {$\emptyseq^{(2)}_{21}$} (q_2) 
          (q_3p) edge node[below] {$\emptyseq^{(2)}$} (q_3) 
          (q_2) edge[bend right] node[right] {$\emptyseq^{(2)}_{21}$} (q_1) 
        ;
      \end{automaton} 
      \caption{Unambiguous automaton $\aut B''$ obtained from $\aut B$ (\exref{ex:ambiguous-auto}).}
      \label{f:refactored-auto}
    \end{minipage}
    \hfill
    \begin{minipage}{0.45\textwidth}
      \centering
      \begin{automaton}
        \node[state,initial]       (S_0) {$S_0$};
        \node[state]                 (S_1) [right=of S_0] {$S_1$};
        \node[state,accepting] (S_f) [below=of S_1] {$S_f$}; 
        \node[state]                 (S_2) [below=of S_0] {$S_2$}; 
        \path[->]
          (S_0) edge node[above] {$a^{12}_{12}$} (S_1)
                  edge node[left]     {$b^{12}_{12}$} (S_2)
          (S_1) edge node[right]   {$\emptyseq^{(2)}$} (S_f) 
          (S_2) edge node[below] {$\emptyseq^{(2)}$} (S_f) 
        ;
      \end{automaton} 
      \caption{Automaton $\aut B_\mathrm{d}$ obtained from $\aut B''$ by \algref{alg:powerset}.}
      \label{f:refactored-auto-det}
    \end{minipage}
  \end{figure}
	\figref{f:refactored-auto} shows the unambiguous automaton $\aut B''$ obtained from 
	$\aut B$ in \exref{ex:ambiguous-auto} using the procedure sketched above. The ambiguous
	transitions $(q_0,b^{12}_{21},q_2)$ and $(q_0,b^{12}_{12},q_3)$ of $\aut B$ have 
	been split into four 
	transitions with new intermediate states $q_2'$ and $q_3'$. \algref{alg:powerset} now
	produces the DFA $\aut B_\mathrm{d}$ shown in 
	\figref{f:refactored-auto-det} when applied to $\aut B''$.
\end{example}

\subsection{Selecting Promising Transitions\label{sec:select}}
Let us now consider the DFA $\detaut A$ that has 
been produced by \algref{alg:powerset}, and use it for recognizing a valid graph
$G\in\L_G(\detaut A)$ using graph moves. First note that $\detaut A$ does not have any 
blank transitions except those that go into the final state, which does
not have any outgoing transitions at all. Consequently, blank transitions are 
only applied when the remaining input is isomorphic to the corresponding blank 
graph, that is, there is no choice when such transitions have to be applied. 
Despite determinism, this is not necessarily the case for the other transitions.
Consider the situation where we have reached a
configuration $(s,G)$ and $G$ contains some edges. 
Although $\detaut A$ is deterministic, several transitions may be applicable because
a graph can be composed from basic graphs in different ways (see \exsref{ex:composition} 
and~\ref{ex:star-auto}). So we are again forced to choose and possibly to backtrack.
To avoid this inefficient backtracking procedure, 
we would like to identify a unique transition, if it exists, that does not
lead into a dead end, by inspecting only local information. To motivate such a
procedure, let us first consider the simplest non-trivial case where $s$ has two
outgoing transitions, say $\delta,\delta'\in\Delta$ with $\delta\neq\delta'$.
Now assume that we can prove that every edge read by a move using
$\delta$ can never be read by any sequence of moves starting with $\delta'$. It
is then clear that any move using $\delta'$ instead of $\delta$ must inevitably
lead into a dead end, and we know for sure that $\delta$ must be picked for
continuing the recognition process. Moreover, backtracking later and trying
$\delta'$ is meaningless because we already know that it will fail. It is thus crucial
to find out whether an edge read by a move using $\delta$ can also be read by a
sequence of moves starting with $\delta'$. This is discussed next.

Again consider the situation where recognition has reached configuration
$(s,G)$, and two transitions $\delta,\delta'\in\Delta$, $\delta\neq\delta'$,
can be used for two competing moves
\begin{align}
	(s,G) & \gmove_\delta(s', G')\label{e:move-delta}      \\
	(s,G) & \gmove_{\delta'}(s'',G'')\label{e:move-delta'}
\end{align}
with $\delta=(s,a^\phi_\rho,s')$, $\delta'=(s,b^{\phi'}_{\rho'},s'')$,
$G=G_\delta\gcomp G'=G_{\delta'}\gcomp G''$ for two graphs $G_\delta$ and
$G_{\delta'}$ with $\ed G_\delta=\{e\}$, $\ed G_{\delta'}=\{e'\}$,
$G_\delta\iso\sem{a^\phi_\rho}$, and $G_{\delta'}\iso\sem{b^{\phi'}_{\rho'}}$.
The sequences $\phi$ and $\phi'$ 
specify how $e$ and $e'$ must be connected to front nodes of $G$ such that
\eqref{e:move-delta} and \eqref{e:move-delta'} are a valid moves, respectively.
To be more precise, let
$\xi,\xi'\colon[\rank(a)]\parto[\rank(s)]$ be partial functions defined by
\begin{align}
	\xi  & = \{(k,i)\in[\rank(a)]\times[\rank(s)]\mid k=\phi(i)\}.\label{e:def-xi}   \\
	\xi' & = \{(k,i)\in[\rank(b)]\times[\rank(s)]\mid k=\phi'(i)\}.\label{e:def-xi'}
\end{align}
Once more, we use the notation $\phi(i)$ to access the $i$-th number within the sequence~$\phi$.
Thus, $\xi$ and $\xi'$ assign the index of a front node (in the sequence of all front nodes of
$G$) to its index in the sequence of all nodes attached to $e$ or $e'$, respectively. If a node is
attached to $e$ or $e'$ as its $k$-th node, but is not a front node of $G$,
$\xi(k)$ or $\xi'(k)$ are undefined, respectively.

We now assume that $e$ is also read in the competing move \eqref{e:move-delta'}
or a subsequent move. If $e$ is read in \eqref{e:move-delta'}, we have $e=e'$ as
well as $a=b$, and $\xi=\xi'$ follows.

If $e\neq e'$, $e$ must be read later, that is, there must be a sequence of
moves
\begin{align}
	(s'',G'')\gmove_{\detaut A}^\ast(q,H)\gmove_{\delta''}(q',H')\label{e:consume-later}
\end{align}
with $\delta''=(q,a^{\phi''}_{\rho''},q')$, $H=H_{\delta''}\gcomp H'$, and $\ed
	H_{\delta''}=\{e\}$, and $H_{\delta''}\iso\sem{a^{\phi''}_{\rho''}}$.

Some front nodes of $G$ may also be front nodes of $H$ in
\eqref{e:consume-later}. Let us indicate this situation by a partial function
$\mu\colon[\rank(s)]\parto[\rank(q)]$ defined by
\begin{equation}
	\mu = \{(i,j)\in[\rank(s)]\times[\rank(q)]\mid\front_G(i)=\front_H(j) \}\label{e:def-mu}
\end{equation}
Consequently, $\{\front_G(i)\mid i\in\domain\mu\}$ is the subset of those front
nodes of $G$ that are also front nodes of $H$. Similar to $\xi$ and $\xi'$, let
us  define the partial function $\xi''\colon[\rank(a)]\parto[\rank(s)]$ by
\begin{align}
	\xi'' & = \{(k,i)\in[\rank(a)]\times\domain\mu\mid k=\phi''(\mu(i))\}.\label{e:def-xi''}
\end{align}
Function $\mu$ is used to refer to front nodes of $G$ instead of front nodes of $H$.

We now show that we have $\xi=\xi''$, similar to the case $e=e'$. To see this,
consider any node $v$ attached to $e$, that is, there is $k\in[\rank(a)]$
such that $v=\att_G(e,k)$.
\begin{itemize}
	\item If $v$ is not a front node of $G$, $k\notin[\phi]$ follows from
	      $G_\delta\iso\sem{a^\phi_\rho}$, and thus $k\notin\domain\xi$. Now
	      assume that $k\in\domain{\xi''}$ and, hence, $k=\phi''(j)$ as well as
	      $j=\mu(i)$ for some $i\in[\rank(s)]$ and $j\in[\rank(q)]$. This
	      implies $\front_G(i)=\front_H(j)=\att_G(e,k)=v$ in contradiction to
	      $v\notin[\front_G]$. Hence, $k\notin\domain{\xi''}$ as well.
	\item If $v$ is a front node of $G$, there is an index $i\in[\rank(s)]$ such
	      that $v=\front_G(i)$ and $k=\phi(i)$, because of
	      $G_\delta\iso\sem{a^\phi_\rho}$. Consequently, $k\in\domain\xi$ and
	      $i=\xi(k)$. When $e$ is read in the last move (using $\delta''$) in
	      \eqref{e:consume-later}, it must also be a front node of $H$ (otherwise,
	      $H$ could not contain~$v$ as a node because moves cannot turn front 
	      nodes into non-front nodes). Hence,
	      $i\in\domain\mu$ and $v=\front_H(\mu(i))$.
	      $H_{\delta''}\iso\sem{a^{\phi''}_{\rho''}}$ also implies
	      $k=\phi''(\mu(i))$, and thus $k\in\domain{\xi''}$ as well as
	      $i=\xi''(k)$.
\end{itemize}
Note that the definition of $\xi''$ in fact subsumes the definition of
$\xi'$ when using the identity on $[\rank(s)]$ as $\mu$. Suppose we could
compute the set $\Xi(\delta',a)$ of all functions $\xi''$ defined as in
\eqref{e:def-xi''} for any graph $G$ where $\mu$ represents any move sequence
starting with $\delta'$ applied to $G$ and finally reading an $a$-labelled edge
(that is, exactly the situation as in \eqref{e:move-delta'} and
\eqref{e:consume-later}). Then we can summarize the discussion above as follows: If $e$ can
be read also in move \eqref{e:move-delta'} or later in \eqref{e:consume-later},
then $\xi\in\Xi(\delta',a)$. Since $\xi$ and $\Xi(a)$ are independent of the
specific choice of $G$ and $e$, we can state the following observation:

\begin{observation}\label{o:dead-end}%
	Let $\delta,\delta'\in\Delta$, $\xi$, and $\Xi(\delta',a)$ defined as
	described above, $\xi\notin\Xi(\delta',a)$, and $G\in\G_\Voc$ any graph such
	that $\delta$ can be applied to $G$. Then $\delta'$ must not be tried for a
	move because it either cannot be applied to $G$, or its application
	leads inevitably into a dead end eventually.
\end{observation}

In fact, $\Xi(\delta',a)$ can be computed in a rather straightforward way
using \algref{alg:follow}. To simplify things, we use the notation
$\Delta_{q}=\{\delta\in\Delta\mid\exists a\in\GVoc,q'\in Q: \delta=(q,a,q')\}$
for indicating the set of transitions leaving a state~$q$. Let us further define
\begin{align}
	\VFollow(\delta)=\{(a,\xi)\mid \xi\in\Xi(\delta,a)\}.\label{e:def-Follow}
\end{align}
The following lemma states that \algref{alg:follow} computes
$\VFollow(\delta)$.
\begin{algorithm}[tb]
	\caption{Determining $\VFollow(\delta_0)$.}
	\label{alg:follow}
	\Input{Finite automaton $\detaut A = (\GVoc, Q, \Delta, q_0, F)$ produced by \algref{alg:powerset} and a transition $\delta_0\in\Delta$.}
	\Output{Follow set $\VFollow$.}
	let $\delta_0\in\Delta_{s}$\;
	$\VWork\leftarrow\{(\delta_0,\nu)\}$ where $\nu=\{(i,i)\mid i\in[\rank(s)]\}$\;
	$\VDone\leftarrow\emptyset$, $\VFollow\leftarrow\emptyset$\;
	\While{$\VWork \neq \emptyset$\label{a2:while}}{
		select and remove any $(\delta,\mu)$ from $\VWork$\;\label{a2:select-Q}
		add $(\delta,\mu)$ to $\VDone$\label{a2:add-to-Done}\;
		let $\delta=(q,\alpha, q')$\;
		\If{$\alpha\notin\Blanks$}{
			let $\alpha=a^\phi_\rho$\;
			$\xi=\{(i,j)\in[\rank(a)]\times\domain\mu\mid i=\phi(\mu(j))\}$\;\label{a2:compute-xi}
			add $(a,\xi)$ to $\VFollow$\;\label{a2:add-pair}
			$\mu'\leftarrow\{(i,j)\in\domain\mu\times[\rank(q')]\mid \phi(\mu(i))=\rho(j)\}$\label{a2:def-f'}\;
			\ForEach{$\delta'\in\Delta_{q'}$ such that $(\delta',\mu')\notin\VDone$}{
				add $(\delta',\mu')$ to $\VWork$\label{a2:add-p'}
			}
		}
	}
\end{algorithm}

\begin{lemma}\label{l:follow}%
	Let $\Voc$ be a ranked alphabet, $\GVoc$ a finite set of canonical graph
	symbols for $\Voc$, $\detaut A = (\GVoc, Q, \Delta, q_0, F)$ a finite
	automaton produced by \algref{alg:powerset}, and $\delta_0\in\Delta$ a transition. 
	Called with $\detaut A$ and $\delta_0$, \algref{alg:follow} returns 
	$\VFollow(\delta_0)$ for $\detaut A$.
\end{lemma}

\begin{sketch}
	Consider any $\detaut A$ and $\delta_0$  as in the lemma. We first show that
	\algref{alg:follow} terminates. To see this, first note that
	$\domain{\mu'}\subseteq\domain \mu$ and $\mu'(x)\in[\rank(q')]$ for each
	$x\in\domain{\mu'}$ of function $\mu'$ defined in \alineref{a2:def-f'}, and
	thus $\domain{\mu'}\subseteq[\rank(s)]$  for every pair $(\delta',\mu')$
	ever added to $\VWork$ in \alineref{a2:add-p'}. Consequently, only finitely
	many pairs $(\delta,\mu)$ can be added to $\VWork$ and $\VDone$. And because
	no pair is added to $\VWork$ again after it has been selected in
	\alineref{a2:select-Q} (and adding it to $\VDone$ in
	\alineref{a2:add-to-Done}), \algref{alg:follow} terminates.

	Now consider any $a\in\Voc$, $\xi\in\Xi(\delta_0,a)$ as for
	\obsref{o:dead-end}, and any graph $G\in\GG_\Voc$. Consequently, there is a
	move sequence starting with $\delta_0$ and finally reading an
	$a$-labeled edge using some transition $\delta$ such that $\xi$ is defined like $\xi'$
	or $\xi''$ in \eqref{e:def-xi'} and \eqref{e:def-xi''}, respectively, and
	using an appropriate partial function $\mu$. By induction on the move
	sequence, one can show that $\VWork$ eventually will contain $(\delta,\mu)$.
	Consequently, $(a,\xi)$ is added to $\VFollow$ in \alineref{a2:add-pair},
	that is, $(a,\xi)$ is contained in the result of \algref{alg:follow} as
	required.

	For the other direction, consider any pair $(a,\xi)$ in the result of
	\algref{alg:follow}. By induction on the number of iterations of the
	while-loop starting at \alineref{a2:while}, one can show that the set
	$\VDone$ only contains pairs $(\delta,\mu)$ such that there exists a move
	sequence starting with $\delta_0$ and ending with $\delta$ so that $\mu$ is
	exactly as in \eqref{e:def-mu}. Since $(a,\xi)\in\VFollow$, it must have
	been added to $\VFollow$ in \alineref{a2:add-pair} at some point after
	selecting $(\delta,\mu)$ in \alineref{a2:select-Q}. Using \eqref{e:def-xi'}
	and \eqref{e:def-xi''}, one can show that $\xi\in\Xi(\delta_0,a)$ as
	required.\qed
\end{sketch}

We now use \obsref{o:dead-end} for identifying the unique transition (if it exists) that
does not lead into a dead end, provided that $\detaut A$ satisfies certain conditions. To this
end, let us define
\begin{align}
	\VNext(\delta)=(a,\xi)\label{e:def-next}
\end{align}
for each transition $\delta=(s,a^\phi_\rho,s')\in\Delta$ where $\xi$ is defined
as in \eqref{e:def-xi}. Moreover, let $\prec_q \;\subseteq
	\Delta_{q}\times\Delta_{q}$ be such that
\begin{align}
	\delta\prec_q\delta' \quad\text{if}\quad \VNext(\delta')\in\VFollow(\delta)\label{e:def-prec}
\end{align}
for each $q\in Q$. According to \obsref{o:dead-end}, $\delta'\not\prec_q\delta$
then indicates that $\delta'$ may only be tried if $\delta$ cannot be applied.

Consider any state $q\in Q$. If $\prec_q^+$ is irreflexive, one can extend
$\prec_q$, by topological sorting, to a strict total order
$\sqsubset_q\;\subseteq \Delta_{q}\times\Delta_{q}$ such that
$\prec_q\;\subseteq\;\sqsubset_q$. Such an order $\sqsubset_q$ has the following
nice property with respect to \obsref{o:dead-end}:

\begin{lemma}\label{l:order}%
	Let $q\in Q$,  $\sqsubset_q\;\subseteq \Delta_{q}\times\Delta_{q}$ a strict total order such that
	$\prec_q\;\subseteq\;\sqsubset_q$, and $\delta,\delta'\in\Delta_{q}$. Then
	$\delta\sqsubset_q\delta'$ implies $\VNext(\delta)\notin\VFollow(\delta')$.
\end{lemma}

\begin{proof}
	Consider any $q\in Q$ and $\delta,\delta'\in\Delta_{q}$ with
	$\delta\sqsubset_q\delta'$. Then $\VNext(\delta)\in\VFollow(\delta')$
	would imply $\delta'\prec_q\delta$ and thus $\delta'\sqsubset_q\delta$,
	contradicting the irreflexivity of $\sqsubset_q$.
\end{proof}

Such a total strict order $\sqsubset_q$ only exists if $\prec_q^+$ is 
irreflexive, which motivates the following:

\begin{defn}[Transition Selection Property]\label{d:ts}%
  A DFA $\aut A = (\GVoc, Q, \Delta, q_0, F)$ has the \emph{transition 
  selection} (\emph{TS}) \emph{property} if $\prec_q^+$ as defined in 
  \eqref{e:def-prec} is irreflexive for each state $q\in Q$.
\end{defn}

An immediate consequence of \lemmaref{l:order} is the following:

\begin{lemma}\label{l:ts}%
   For every DFA $\aut A$ with the TS property, there is a constant time procedure which,
   given a configuration of $\aut A$, selects the unique transition
   that allows to continue the recognition process without running into a dead end,
   provided that such a transition exists.
\end{lemma}

\begin{proof}
	Assume that a configuration with state $q$ has been reached during
	recognition. There exists an order $\sqsubset_q$ as in \lemmaref{l:order} 
	because of the TS property. Now try all outgoing transitions of $q$ in ascending
	order of $\sqsubset_q$ and pick the first one, say $\delta$, that can be
	applied. Then $\delta\sqsubset_q\delta'$ holds for every outgoing transition
	$\delta'$ that has not yet been tried, and thus
	$\VNext(\delta)\notin\VFollow(\delta')$ by \lemmaref{l:order}, that is,
	$\delta'$ must not be tried by \obsref{o:dead-end} because $\delta$ can be
	applied.
	
	The relation $\sqsubset_q$ depends only on $\aut A$ and $q$, and can thus
	be precomputed for every state $q$.
	Checking whether a given transition applies
	takes constant time, because all that is required is to test
	whether an edge with the relevant label is attached to the front nodes of the graph
	as determined by the symbol to be read by the transition. Using appropriate data structures,
	this can straightforwardly be implemented to run in constant time.
\end{proof}

\begin{example}\label{ex:transition-selection}%
	Only state $S_1$ of the DFA $\aut S'$ shown in 
	\figref{f:det-star-auto} has two outgoing transitions, $\delta_2$ and 
	$\delta_3$. We obtain
	\begin{align*}
		\VNext(\delta_2) & = (a,\{(1,1)\})&
		\VFollow(\delta_2)&=\{(a,\{(1,1)\}),(b,\{(1,1)\})\}\\
		\VNext(\delta_3) & = (b,\{(1,1)\})&
		\VFollow(\delta_3)&=\{(b,\{(1,1)\})\}
	\end{align*}
	by applying \eqref{e:def-next} with \eqref{e:def-xi} and 
	using \algref{alg:follow}, that is, $\delta_2\prec_{S_1}\delta_3$ and 
	$\delta_3\not\prec_{S_1}\delta_2$, that is, $\aut S'$ has the TS
	property. We can thus choose 
	$\delta_2\sqsubset_{S_1}\delta_3$, that is, one must try $\delta_2$ first
	when one has reached $S_1$, and $\delta_3$ only if $\delta_2$ is not
	applicable, that is, if there is no remaining outgoing edge with label $a$.
\end{example}

\subsection{Free Edge Choice\label{sec:FEC}}
Unfortunately, the transition selection property does not necessarily prevent
the recognition process from running into dead ends, even for valid graphs. This is
the case if one has the choice between different edges to be read by a selected
transition, but not all of them are equally suited. This is demonstrated in
the following example:

\begin{example}\label{ex:fec}
  \begin{figure}[tb]
    \begin{automaton}
      \node[state,initial] (S_0) {$S_0$};
      \node[state] (S_1) [right=of S_0] {$S_1$};
      \node[state](S_2) [right=of S_1] {$S_2$};
      \node[state,accepting](S_f) [right=of S_2] {$S_f$};
      \path[->]
        (S_0) edge[bend angle=45,bend left] node[above] {$a^1_{12}$}
                        node[below,text=\Red] {$\delta_1$} (S_1)
        (S_1)  edge[bend angle=45,bend left] node[below] {$b^{31}_3$}
                        node[above,text=\Red] {$\delta_2$} (S_0)
                 edge node[above] {$c^{31}_2$}
                        node[below,text=\Red] {$\delta_3$} (S_2)
        (S_2) edge node[above] {$\emptyseq^{(1)}$}
                        node[below,text=\Red] {$\delta_4$} (S_f)
        ;
    \end{automaton}
    \hfill
    \begin{graph}[x=10mm,y=-5mm]
      \inode(v)(2,0) 
      \inode(w)(3,0) 
      \fnode(x)(1,1) 
      \inode(y)(2,2) 
      \rnode(z)(3,2) 
      \path
      (v) edge[->] node[above] {\small$b$} (w)
      (x) edge[->] node[above] {\small$a$} (v)
           edge[->] node[below] {\small$a$} (y)
      (y) edge[->] node[below] {\small$c$} (z)
      ;
    \end{graph}
    \hfill
    \begin{graph}[x=10mm,y=-5mm]
      \inode(v)(2,0) 
      \inode(w)(3,0) 
      \fnode(x)(1,1) 
      \Fnode(y)(2,2){1}
      \rnode(z)(3,2) 
      \path
      (v) edge[->] node[above] {\small$b$} (w)
      (x) edge[->] node[above] {\small$a$} (v)
      (y) edge[->] node[below] {\small$c$} (z)
      ;
    \end{graph} \\[+2.7mm]
    \caption{Finite automaton $\aut F$ and two graphs used in \exref{ex:fec}.}
    \label{f:fec}
  \end{figure}
  Let $a$, $b$, and $c$ be symbols of rank 2. 
  The DFA $\aut F$ shown in \figref{f:fec} always
  allows to select a unique transition: $\delta_1$ in state $S_0$,
  $\delta_2$ in $S_1$ if the second front interface node has an
  outgoing $b$-edge, and $\delta_3$ if it has an outgoing
  $c$-edge. The left graph in \figref{f:fec} can thus be recognized by
  first reading the edge from $x$ to $v$ in state $S_0$ and then
  following the uniquely determined moves.

  However, the recognition process will run into a dead end if it
  reads the edge from $x$ to $y$ first. This move then yields a
  configuration $(S_1,G')$ where $G'$ is the right graph in
  \figref{f:fec}. The second front interface node ($y$) has an
  outgoing $c$-edge, but no $b$-edge, that is, $\delta_3$ is
  selected. This transition tries to find a composition
  $G'=C\gcomp G''$ with $C\iso\sem{c^{31}_2}$, that is, $\front_C=xy$,
  $\rear_C=z$, and thus $x\notin\nd G''$, which would leave the edge
  from $x$ to $v$ dangling. Hence, $\delta_3$ cannot be applied which
  leaves the recognition process stuck in a dead end.\qed
\end{example}

Such a situation cannot happen if every edge that can be chosen to be
read next allows to continue the recognition process until an
accepting configuration is reached, or none of them does. We say that the next edge can be
freely chosen:

\begin{defn}[Free Edge Choice Property]\label{d:fec} A finite automaton $\aut A
		= (\GVoc, Q, \Delta, q_0, F)$ has the \emph{free edge choice} (\emph{FEC})
	\emph{property} if $(q',H')\in\CorrectConf{\aut A}$ implies
	$(q',H'')\in\CorrectConf{\aut A}$ for all move sequences
	$(q_0,G)\gmove_{\aut A}^\ast(q,H)\gmove_\delta(q',H')$ and
	$(q,H)\gmove_\delta(q',H'')$ where $\delta\in\Delta$.%
	\footnote{Recall that $\CorrectConf{\aut A}$ is the set of all 
		acceptable graph configurations of~$\aut A$, that is, those configurations from 
		where an accepting configuration can be reached (see \defref{d:graph-conf}).}
\end{defn}
	
By the discussion above, the recognition process using an
automaton with the TS property and the FEC
property will never run into a dead end when applied to a valid graph. It
remains to show how the FEC property can be checked for a given automaton. We present
a sufficient condition that is easy to test. Thus, an automaton is guaranteed to
have the FEC property if it passes the test, but the converse is not necessarily
true.

Consider a situation during the recognition of a valid graph when a transition
has been selected for the next move, and there are several edges one can choose
from. All the edges that are not being chosen now must be read later during the
recognition process. This is only possible for certain transitions. Let
us call these transitions \emph{deferrable} because one edge can be read now
whereas reading of the others is deferred:

A transition $\delta\in\Delta$ is called \emph{deferrable} if and only if there
are move sequences
\begin{align}
	                & (q,G)\gmove_\delta(q',G') \label{e:defer1}                                                \\
	\text{and}\quad & (q,G)\gmove_\delta(q',G'')\gmove_{\aut A}^\ast(s,H)\gmove_{\aut A}(s',H')\label{e:defer2}
\end{align}
such that the edge read in \eqref{e:defer1} and the one read in the last move of
\eqref{e:defer2} are the same, that is $\ed G'\setminus\ed G=\ed H'\setminus\ed
	H$.

One can easily check with an algorithm similar to \algref{alg:follow} whether a
transition is deferrable. With this information, one can then check
whether $\aut A$ has the FEC property:

\begin{lemma}\label{l:fec}%
	A finite automaton $\aut A$ has the free edge choice property if
	$[\rho]\subseteq[\phi]$ holds for every transition $(q,a^\phi_\rho,q')$ of
	$\aut A$ that is deferrable.
\end{lemma}

\begin{proof}
	Consider a finite automaton $\aut A = (\GVoc, Q, \Delta, q_0, F)$ that
	satisfies the condition in the lemma. Let us assume that $\aut A$ does not
	have the free edge choice property, that is, there are move sequences
	$(q_0,G)\gmove_{\aut A}^\ast(q,H)\gmove_\delta(q',H')\in\CorrectConf{\aut
			A}$ and $(q,H)\gmove_\delta(q',H'')\notin\CorrectConf{\aut A}$. The two
	moves using $\delta$ must have read two different edges, say $e$ and $e'$,
	that is, $H=F\gcomp H'=F'\gcomp H''$, $\ed F=\{e\}$, and $\ed F'=\{e'\}$.
	$H'$ then still contains $e'$, which is read later, because
	$(q',H')\in\CorrectConf{\aut A}$. Consequently, there is a move sequence
	$(q',H')\gmove_{\aut A}^\ast(s,I)\gmove_{\aut A}(s',I')\in\CorrectConf{\aut
			A}$ such that $\ed I\setminus\ed I'=\{e'\}$, that is, $\delta$ is
	deferrable. Now let $\delta=(q,a^\phi_\rho,q')$. We have
	$a=\lab_G(e)=\lab_G(e')$ since $e$ and $e'$ can both be read by $\delta$. And we
	have $[\rho]\subseteq[\phi]$ since $\delta$ is deferrable, and $\aut A$
	satisfies the condition in the lemma. Note that
	$\front_H=\front_F=\front_{F'}$. Thus $[\rho]\subseteq[\phi]$ implies
	$\rear_F=\rear_{F'}=\front_{H'}=\front_{H''}$. Consequently, all nodes of
	$H$ that are attached to $e$ or $e'$, but that are not in $\front_H$, cannot
	be attached to any other edge, and thus $H'\iso H''$, where the isomorphism
	maps $e$ to $e'$ (and vice versa). Then $(q',H'')\in\CorrectConf{\aut A}$
	follows from $(q',H')\in\CorrectConf{\aut A}$, contradicting the assumption.
	$\aut A$ hence has the free edge property.
\end{proof}

\begin{example}
Transition $\delta_1$ is the only deferrable transition of automaton $\aut F$ in 
\figref{f:fec}, but it violates the condition of \lemmaref{l:fec} 
as it reads $a^1_{12}$, and $[\rho]=\{1,2\}\not\subseteq\{1\}=[\phi]$.
Consequently, $\aut F$ does not pass the FEC test, which leaves the question open 
whether $\aut F$ has the FEC property or not. In fact, we have already 
seen in \exref{ex:fec} that $\aut F$ does not have the FEC property.

The automaton $\aut S'$ in \figref{f:det-star-auto}, however, has the FEC property
according to \lemmaref{l:fec} since $\delta_2$ is its only deferrable 
transition, and $[\rho]=\{1\}\subseteq\{1\}=[\phi]$ for the graph symbol
$a^1_1$ read by~$\delta_2$. Since $\aut S'$ also has the TS 
property (see \exref{ex:transition-selection}), one can use $\aut S'$ for recognizing
star graphs without any backtracking.
\end{example}

\subsection{Results}\label{s:results}%
The following theorem summarizes the findings of this paper.

\begin{theorem}\label{t:mainresults}
\begin{enumerate}
\item Every finite automaton over the canonical alphabet can effectively be turned into a
DFA.\footnote{As usual, and unavoidably, the powerset
construction can cause an exponential blow-up in the size of the automaton.}
\item Given a DFA $\aut A$ (of the type produced by
\algref{alg:powerset}), it can be decided in polynomial time (in the number 
of transitions of $\aut A$) whether $\aut A$
has the TS and the FEC property. 
\item If $\aut A$ has the TS and the FEC property, the membership problem for $\L(\aut A)$
can be decided in linear time.
\end{enumerate}
\end{theorem}

\begin{proof}
The first statement follows from \thmref{thm:deterministic} and the previously
mentioned fact that ambiguous automata can be made unambiguous. 
Deciding the TS property requires polynomial time because  \algref{alg:follow} 
and topological sorting must be used for checking irreflexivity of $\prec_q^+$
for every state $q$ of $\aut A$. The second statement then
follows from \lemmaref{l:fec} and the fact that it is decidable in polynomial time
whether a transition is deferrable (see above). Finally, to decide whether a graph
belongs to $\L(\aut A)$, we can repeatedly apply transitions using \lemmaref{l:ts}.
\end{proof}

The concepts described in this paper have been realized in a graphical tool. It
reads in a finite automaton over the canonical alphabet, or creates it from a regular expression in the
obvious way, turns it into a DFA, and uses it for
recognizing input graphs if the DFA has the TS and the FEC
property. \figref{f:tool} shows a screenshot with an input graph in the
top-left window, and the DFA in the top-right window after the DFA has recognized
the input graph. The bottom window shows how the input graph has been composed
from basic graphs such that it has been accepted by the DFA. 
\begin{figure}[t]
\centering\includegraphics[width=\textwidth]{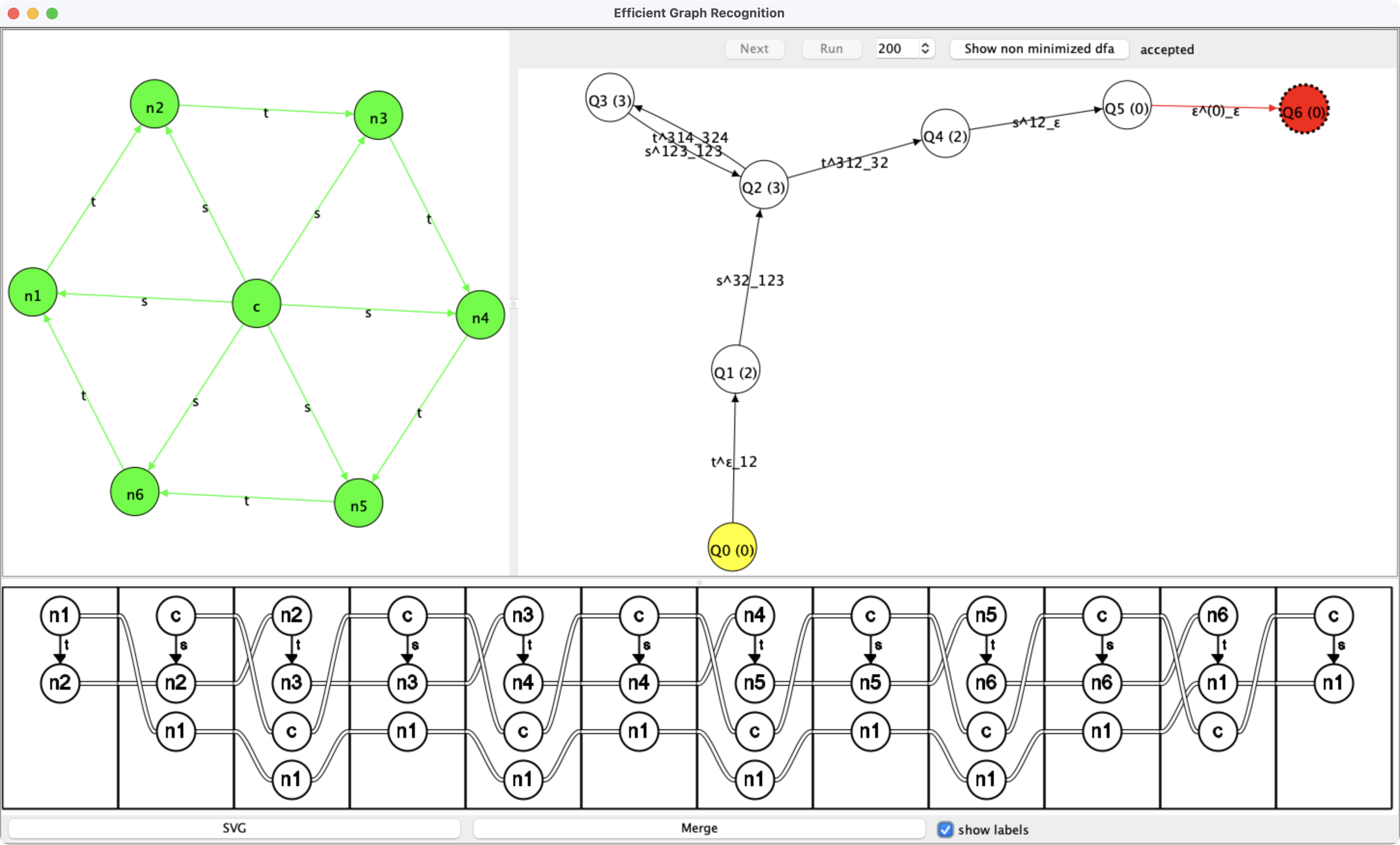}
\caption{Screenshot of a tool realizing efficient graph recognition based on 
         finite automata.}
\label{f:tool}
\end{figure}

\tableref{t:examples} shows some example graph languages that can be specified
and efficiently recognized by finite automata: \emph{Stars} is the language of
stars used in this paper (\exref{ex:star-auto}), \emph{Wheels} contains wheel
graphs introduced in \cite[p.~92]{Habel:92}, \emph{Palindromes} contains all
palindromes over $\{a,b\}$, and $a^nb^nc^n$ the language $\{a^nb^nc^n\mid n>0\}$
as string graphs~\cite[Ex.~2.4]{Habel:92}. \tableref{t:examples} shows the
type of the contained graphs, that is, the size of their front and rear
interfaces, and the size of their DFA in terms of the number of states and
transitions. Moreover, \emph{max.~\#nodes} shows the number of nodes of 
the largest atom used in the DFA. Each of the automata has the TS and FEC property.
\begin{table}[t]
\caption{Example graph languages and their finite automata.}
\label{t:examples}
\centering
\begin{tabular}{lcccccc}
\toprule
Language & Type & \#states & \#transitions & max.~\#nodes & TS & FEC \\ \midrule
Stars & $(1,1)$ & 4 & 4 & 2 & $\checked$ & $\checked$ \\
Wheels & $(1,0)$ & 7 & 7 & 4 & $\checked$ & $\checked$ \\
Palindromes & $(2,0)$ &  5 & 7 & 3 & $\checked$ & $\checked$ \\
$a^nb^nc^n$ & $(2,0)$ & 9 & 10 & 4 & $\checked$ & $\checked$ \\\bottomrule
\end{tabular}
\end{table}

\figref{f:tool} shows the DFA for wheels and a wheel with six
spokes (\textsf{s}-edges for spokes and \textsf{t}-edges for the tread of 
the wheel). The front and rear interfaces of the wheel are empty. As shown in 
the bottom window, the left-most atom refers to the \textsf{t}-edge from 
\textsf{n1} to \textsf{n2}, and recognition of the wheel has 
started with this part of the tread of the wheel. Recognition could 
have started with any of the \textsf{t}-edges because the DFA has the FEC 
property. However, the FEC test fails for this DFA because of the 
first transition from \textsf{Q0} to \textsf{Q1} with label 
\textsf{t\^{}$\emptyseq$\_12}, which represents $t^\emptyseq_{12}$. This 
transition is deferrable, yet $\{1,2\}\not\subseteq\emptyset$. Hence,
\emph{Wheels} is an example where the simple FEC test fails although the
automaton does have the FEC property.

\section{Conclusions}
\label{s:concl}
In this paper we have defined finite automata for graph recognition,
devised a modified powerset construction to make them deterministic,
and stated two criteria (transition selection and free edge choice)
under which these automata work without backtracking.

This paper leaves some questions open.
Do (nondeterministic)
finite graph automata recognize NP-complete languages, which exist for
(context-free) hyperedge replacement
grammars?
How do the languages accepted by finite graph automata relate to those
definable with the ``regular''  graph grammars of Gilroy \emph{et al.}
\cite{Gilroy-Lopez-Maneth:17}, which have an efficient parsing
algorithm (not defined with automata) that also requires free edge
choice? (Probably these language classes are incomparable.)

We intend to continue our work in several directions.
First, it is easy to define graph expressions defining graph languages
by composition, Kleene star, and union. These expressions could be used
in the nested graph conditions of Habel and Pennemann
\cite{Habel-Pennemann:05,Pennemann:09} to specify global graph
properties \cite{Gaifman:82}, and check them with automata.
Second, we could consider borrowed nodes in graph expressions and
graph automata that can be contracted with other nodes in a graph \cite{Drewes-Hoffmann-Minas:22}.
Then these mechanisms allow to define graphs of unbounded
treewidth. 

\paragraph{Acknowledgments.}
	 We are very grateful to the anonymous reviewers for their 
	 helpful comments and suggestions, which have helped to 
	 improve the quality of this paper.

\bibliographystyle{eptcs} 
\bibliography{References}

\end{document}

%% file: myeptcstheoremdefs.tex
\newtheoremstyle{plainbreak}%
  {}{}%
  {\itshape}{}%
  {\bfseries}{}
  {\newline}{}
\newtheoremstyle{definitionbreak}%
  {}{}%
  {}{}%
  {\bfseries}{}
  {\newline}{}
\newtheoremstyle{remarkbreak}%
  {}{}%
  {}{}%
  {\itshape}{}
  {\newline}{}
\theoremstyle{definition}
  \newtheorem{defn}{Definition}[section]

  \newtheorem{example}{Example}[section] 

\theoremstyle{remark}

  \newtheorem*{sketch}{Proof Sketch}
\theoremstyle{plain}
  \newtheorem{theorem}[defn]{Theorem}
  \newtheorem{lemma}[defn]{Lemma}

  \newtheorem{observation}[defn]{Observation}
  \newtheorem{fact}[defn]{Fact}
\theoremstyle{definitionbreak}
  
\theoremstyle{plainbreak}



%% file: mymath.tex
\newcommand{\domain}[1]{\mathit{dom}(#1)}

\def\LL(#1){\ensuremath{\mathit{LL}(#1)}}
\def\LR(#1){\ensuremath{\mathit{LR}(#1)}}
\def\LALR(#1){\ensuremath{\mathit{LALR}(#1)}}
\def\SLR(#1){SLR(#1)}
\def\SLL(#1){SLL($#1$)}
\def\SLRo(#1){\ensuremath{\mathit{SLR}^\bullet\!(#1)}}

\def\Fi#1(#2){\mathit{First}_{#1}(#2)}
\def\Fo#1(#2){\mathit{Follow}_{#1}(#2)}

\def\FiFo#1(#2,#3){\Fi_{#1}(#3) \cdot_{#1} \Fo_{#1}(#2)}
\def\FIFO#1(#2,#3){\mathit{FiFo}_{#1}(#3,#2)}

\newcommand{\type}{\mathit{type}}

\newcommand{\iso}{\cong}

\newcommand{\move}{\mathrel\vdash}

\newcommand{\tf}[1]{\textsf{#1}}

\newcommand{\yields}[1][]{%
  \def\temp{#1}%
  \ifx\temp\empty\operatorname{::=}\else%
    \operatornamewithlimits{::=}\limits_{\tf{#1}}%
  \fi}

\newcommand{\card}[1]{| #1 |}              
\newcommand{\tup}[1]{\langle #1 \rangle}   
\newcommand{\emptyseq}{\varepsilon}        
\newcommand{\Voc}{\Sigma}       
\newcommand{\GVoc}{\Theta}       
\newcommand{\nd}[1]{\dot{#1}}   
\newcommand{\ed}[1]{\bar{#1}}   
\newcommand{\att}{\mathit{att}}   

\newcommand{\rank}{\mathit{rank}}   
\newcommand{\lab}{\mathit{lab}}               
\newcommand{\front}{\mathit{front}}               
\newcommand{\rear}{\mathit{rear}}               
\newcommand{\Nat}{\mathbb{N}}

\newcommand{\aut}[1]{\mathfrak{#1}}

\newcommand{\G}{\mathcal{G}}

\renewcommand{\L}{{\mathcal{L}}}
\renewcommand{\LL}{{\mathbb{L}}}

\makeatletter
\def\tdcfg#1#2{\@ifnextchar[{\@tdcfg{#1}{#2}}{({#1},{#2})}}
\def\@tdcfg#1#2[#3]{({#3},{#1},{#2})} 
\def\cfG(#1)(#2)(#3){%
  \!\left\lceil\frac{#1}{}\!\text{\scriptsize$\circ$}\!\frac{#2}{#3}\right\rceil
}
\makeatother

\newcommand{\CDOT}{{\,\centerdot\,}}

\newcommand{\xyrightarrow}[3]{\mathop{\!\!\!\xymatrix@C=#1{{}\ar@{#2}[r]_{#3}&{}}\!\!\!}\nolimits}
\newcommand{\xyRightarrow}[3]{\mathop{\!\!\!\xymatrix@C=#1{{}\ar@<1pt>@{#2}[r]\ar@<-1pt>@{#2}[r]_{#3}&{}}\!\!\!}\nolimits}
\newcommand{\xyarrow}[2]{%
  \sbox{0}{$\scriptstyle#2$}%
  \mathop{\!\!\!\xymatrix@C\dimexpr\wd0+4pt\relax{{}\ar@{#1}[r]_{#2}&{}}\!\!\!}%
}

\newcommand\Input[1]{\textit{input}({#1})}

\newcommand{\sem}[1]{\left\llbracket #1 \right\rrbracket}
\newcommand{\gcomp}{\mathop{\odot}}

\def\buc(#1,#2,#3){[#1] #2 \CDOT #3}

\def\epsilon{\varepsilon}
\def\emptyset{\varnothing}
\def\theta{\vartheta}
\def\rho{\varrho}
\def\phi{\varphi}

\def\autoconf(#1,#2,#3){#3 {\scriptstyle\lozenge} [#1]^{#2}}

\def\GG{\mathbb{G}}

\def\pstate(#1,#2){\langle #1, #2 \rangle}

\def\cfaconf(#1,#2){#1 {\scriptstyle\blacklozenge} #2}

\newcommand{\idblank}[1]{\emptyseq^{(#1)}}
\newcommand{\VFollow}{\mathit{follow}}
\newcommand{\VNext}{\mathit{next}}
\newcommand{\VWork}{\mathit{frontier}}
\newcommand{\VDone}{\mathit{done}}

\newcommand{\assuref}[1]{\hyperref[#1]{Assumption~\ref*{#1}}}
\newcommand{\Blanks}{\mathcal{B}}
\newcommand{\closure}[1]{\mathop{\textit{Cl}}(#1)}

\newcommand{\Terms}{\GVoc}
\newcommand{\AllTerms}{\mathbf{\GVoc}}

\newcommand{\gmove}{\mathrel\Vdash}
\newcommand{\parto}{\rightharpoonup}
\newcommand{\detaut}[1]{\aut{#1}_\mathrm{d}}
\newcommand{\gatom}[3]{\tup{#1}^{#2}_{#3}}
\newcommand{\gblank}[2]{\tup{\emptyseq}^{(#1)}_{#2}}
\newcommand{\sblank}[2]{\emptyseq^{(#1)}_{#2}}
\newcommand{\CorrectConf}[1]{\mathcal C_{#1}}

\DontPrintSemicolon
\SetKwProg{Proc}{procedure}{}{}
\SetKwInOut{Input}{Input}
\SetKwInOut{Output}{Output}
\SetAlgoHangIndent{1em}
\SetNlSty{textsf}{}{}


%% file: mytikz.tex
\pgfdeclarelayer{background}              
\pgfdeclarelayer{middleground}
\pgfdeclarelayer{foreground}
\pgfsetlayers{background,middleground,main,foreground}

\newcommand\Red{red!75!black}

\newcommand\BackgroundColor{gray!15!white}
\tikzset{
  x=8mm,y=8mm,>=latex,            
  background rectangle/.style
  ={fill=\BackgroundColor,rounded corners=4pt
      },
  glass/.style ={opacity=0,text opacity=0},     
  satin/.style ={opacity=0.3,text opacity=0.3},
  e/.style={inner sep=0.0pt,minimum size=0pt},
  o/.style={circle,draw,fill=white,font=\scriptsize,inner sep=1pt,minimum size=2.5mm},
  }
  {\begin{tikzpicture}%
    [x=6mm,y=-4mm, 
    #1]}%
  {\end{tikzpicture}}
\newenvironment{automaton}[1][x=10mm,y=10mm]%
{\begin{tikzpicture}[inner frame sep=4pt,show background rectangle,
    baseline=(current bounding box.center),
    node distance=8mm,
    initial text=,
    every state/.style={fill=white,minimum size=1mm,inner sep=2pt},
    every label/.style={node id},
    #1]}%
  {\end{tikzpicture}}
\newenvironment{graph}[1][x=10mm,y=10mm]%
{\begin{tikzpicture}[%
    inner frame xsep=0pt,inner frame ysep=4pt,show background rectangle,
    baseline=(current bounding box.center),
    every label/.style={node id},
    #1]}%
  {\end{tikzpicture}}
\def\frontptr(#1){\path (f-#1) edge[-,double distance=1pt,double] (#1)}
\def\rearptr(#1){\path (r-#1) edge[-,double distance=1pt,double] (#1)}
\def\invnode(#1)(#2){
  \node[o,glass] (#1) at (#2) {$#1$};
}
\def\inode(#1)(#2){
  \node[o] (#1) at (#2) {$#1$};
}
\def\fnode(#1)(#2){
  \inode (#1)(#2)
  \node[e] (f-#1) at ($(#2)-(0.5,0)$) {};
  \frontptr(#1);
}
\def\Fnode(#1)(#2)#3{
  \inode (#1)(#2)
  \node[e] (f-#1) at ($(#2)-(#3,0)-(0.5,0)$) {};
  \frontptr(#1);
}
\def\rnode(#1)(#2){
  \inode (#1)(#2)
  \node[e] (r-#1) at ($(#2)+(0.5,0)$) {};
  \rearptr(#1);
}
\def\frnode(#1)(#2){
  \fnode(#1)(#2)
  \node[e] (r-#1) at ($(#2)+(0.5,0)$){};
  \rearptr(#1);
}
  {\begin{tikzpicture}%
     [>=Computer Modern Rightarrow                 
     ,every node/.style={rectangle,inner sep=3pt}
     ,incl/.tip={Hooks[round,right]}               
     ,incL/.tip={Hooks[round,left]}
     ,part/.tip={Circle[open]} 
     ,weak/.tip={Latex[open]} 
     ,harp/.tip={Computer Modern Rightarrow[left]} 
     ,Harp/.tip={Computer Modern Rightarrow[right]}
     ,to/.style={->}
     ,ito/.style={incl->}
     ,ifrom/.style={incL->}
     ,eto/.style={->>}
     ,mto/.style={>->}
     ,pto/.style={part->}
     ,peto/.style={part->>}
     ,pmto/.style={part>->}
     ,tow/.style={-weak}
     ,etow/.style={->weak}
     ,mtow/.style={>-weak}
     ,ptow/.style={part->weak}
     ,petow/.style={part->weak}
     ,pmtow/.style={part>-weak}
     ,To/.style={double equal sign distance,-Implies}
     ,x=12mm,y=12mm
     ,#1]}%
  {\end{tikzpicture}}
  {\begin{tikzpicture}[#1]}%
  {\end{tikzpicture}}
\makeatletter

\def\DownEdge[#1]{\tikz \draw (0pt,8pt) edge[#1] (0pt,0pt);}

\def\blob(#1) at (#2)#3{%
  \node[blub](#1)at(#2) {\begin{tabular}{@{}c@{}}  #3 \end{tabular}
};
  }
\def\enlargebb{\node [glass,fit= (current bounding box),inner sep=2pt] {};}

\def\Highlight#1#2#3(#4){
  \begin{pgfonlayer}{background} 
      \node (#4) [subgraph,fit= #1,inner sep=#3,fill=#2] {}; 
   \end{pgfonlayer}
}

\def\Just#1%
  {\tikz \draw (0pt,0pt) edge[#1] (12pt,0pt); 
   \end{graph}}
\def\Graph{\@ifnextchar[{\@Graph}{\@GrapH}}
\def\@Graph[#1]#2{%
  \BOX{%
    \begin{tikzpicture}[x=8mm,y=8mm,label distance=-2pt,>=latex,#1]
     #2%
    \end{tikzpicture}%
  } 
}
\def\@GrapH#1{%
  \BOX{%
    \begin{tikzpicture}[x=8mm,y=8mm,label distance=-2pt,>=latex]
     #1%
    \end{tikzpicture}%
  } 
}

\def\proGraph{\@ifnextchar[{\@proGraph}{\@pro@Graph}}
\def\@proGraph[#1]#2{\BOX{%
  \begin{tikzpicture}[x=8mm,y=-7mm,>=latex,every label/.style={elab},#1]
     #2
    \enlargebb
   \end{tikzpicture}}}
\def\@pro@Graph#1{\BOX{%
  \begin{tikzpicture}[x=8mm,y=-7mm,>=latex,every label/.style={elab}]
    #1
    \enlargebb
  \end{tikzpicture}}}

\def\uvar(#1)#2{\@ifnextchar[{\@uvar(#1)#2}{%
  \node (#1-node) at ($(#1)+(0,1)$) {#2};
  \path[arm] (#1) edge (#1-node);
}}
\def\@uvar(#1)#2[#3]{%
  \node (#1-node) at ($(#1)+(#3,1)$) {#2};
  \path[arm] (#1) edge (#1-node);
}
\def\ustar#1#2{\@ifnextchar[{\@ustar{#1}{#2}}{\Graph{%
      \node (r) [term] at (1,1) {#1};
      \node (h) [nont] at (1,2) {#2};
      \path (r) edge[arm] (h);
}}}
\def\@ustar#1#2[#3]{\Graph[#3]{%
      \node (r) [term] at (1,1) {#1};
      \node (h) [nont] at (1,2) {#2};
      \path (r) edge[arm] (h);
}}

\def\custar#1#2#3{\@ifnextchar[{\@custar(#1)#2}{\proGraph{%
      \node (r) [term] at (1,1) {#1};
      \node (h) [nont] at (1,2) {#2};
      \path (r) edge[arm] (h);
      \node (c) [term] at (1,3) {#3};
}}}
\def\@custar#1#2#3[#4]{\proGraph[#4]{%
      \node (r) [term] at (1,1) {#1};
      \node (h) [nont] at (1,2) {#2};
      \path (r) edge[arm] (h);
      \node (c) [term] at (1,3) {#3};
}}
\makeatother

\newcommand{\BOX}[1]{\begin{array}{@{}c@{}}#1\end{array}}
